\PassOptionsToPackage{square, numbers, sort}{natbib}
\documentclass[%
reprint,
 amsmath,
 amssymb,
 aps,
]{revtex4-2}

\usepackage{graphicx}
\graphicspath{ {./pictures/} }
\usepackage{placeins}
\usepackage{csquotes}
\usepackage{enumitem}
\usepackage{amsmath}
\usepackage{amsfonts}
\usepackage{amssymb}
\usepackage{amsthm} %%% for changing the styles of math enviroments
\usepackage{tcolorbox} %%% for creating color boxes e.g. enviroments
\usepackage{mathtools}
\usepackage{nicefrac}
\usepackage{bm}
\usepackage{physics}%%% norm symbol
\usepackage{esint} %%% for other integrals
\usepackage[title]{appendix}

\newcommand{\ii}{\mathrm{i}}
\newcommand{\ee}{\mathrm{e}}

\newcommand{\function}[3]{#1: #2 \rightarrow #3}
\newcommand{\R}{\mathbb{R}}

\newcommand{\intd}{\;\mathrm{d}}
\newcommand{\C}{\mathbb{C}}

\usepackage{xcolor}
\usepackage[normalem]{ulem}

\theoremstyle{plain}
\newtheorem{theorem}{Theorem}[section]

\theoremstyle{plain}

\theoremstyle{plain}

\theoremstyle{definition}
\newtheorem{definition}[theorem]{Definition}

\theoremstyle{definition}

\theoremstyle{definition}

\theoremstyle{remark}
\newtheorem*{remark}{Remark}
\tcolorboxenvironment{theorem}{
	colback=blue!5!white,
	boxrule=0pt,
	boxsep=1pt,
	left=2pt,right=2pt,top=2pt,bottom=2pt,
	oversize=2pt,
	sharp corners,
	before skip=\topsep,
	after skip=\topsep,
}

\tcolorboxenvironment{definition}{
	colback=yellow!5!white,
	boxrule=0pt,
	boxsep=1pt,
	left=2pt,right=2pt,top=2pt,bottom=2pt,
	oversize=2pt,
	sharp corners,
	before skip=\topsep,
	after skip=\topsep,
}

\tcolorboxenvironment{corollary}{
	colback=red!5!white,
	boxrule=0pt,
	boxsep=1pt,
	left=2pt,right=2pt,top=2pt,bottom=2pt,
	oversize=2pt,
	sharp corners,
	before skip=\topsep,
	after skip=\topsep,
}

\makeatletter
\renewcommand*{\@fnsymbol}[1]{%
\ensuremath{\ifcase#1 \or  \or  \or * \or \dagger \or \ddagger \or
   \mathsection \or \mathparagraph \or \| \or ** \or \dagger\dagger
   \or \ddagger\ddagger \else \@ctrerr\fi}}
\makeatother

\begin{document}

\preprint{APS/123-QED}

\title{Curvature-induced bound states in quantum wires}

\author{Tim Bergmann$^{\ast}$}
 \thanks{$\ast$ These authors contributed equally to the work.}
\author{Benjamin Schwager$^{\ast \dagger}$}
 \email{$\dagger$ benjamin.schwager@physik.uni-halle.de}
\author{Jamal Berakdar}

\affiliation{Institut für Physik, Marin-Luther-Universität Halle-Wittenberg, 06099 Halle, Germany }

\date{\today}

\begin{abstract}
    A classical particle under spatial constraints is strictly confined to live on a specific space manifold or path, but this assumption is incompatible with the zero-point fluctuations of a quantum particle. One way to describe quantum mechanics under constraints is the confinement potential approach (CPA). For a non-relativistic particle, the CPA maps the problem onto the solution of a Schrödinger-type equation in an isometrically embedded Riemannian submanifold of Euclidean space while the motion along orthogonal directions are decoupled and spatially confined. This approach respects quantum uncertainty, and one of its key results is the appearance of geometry- and metric-induced potentials that affect the stationary states and the dynamics of the particle. For particles constrained to different spaces, such as structures hosting sharp bents, vertices, wedges, conical apices, tips, or self-intersections, a formalism beyond the CPA is needed. Here, a step towards a CPA extension for irregular spaces is presented. After classifying the possible geometric irregularities concerning the CPA formalism, the presentation is focused on a sharply bent quantum wire modeled as an embedded curve with singular (but absolute integrable) curvature. For a subclass fulfilling the additional requirement that the geometric potential is a distribution of first order, a solution scheme for the confined Schrödinger equation is presented based on singular Sturm-Liouville theory and operator theoretic methods. The analytical considerations and numerical simulations evidence the existence of curvature-induced bound states with non-differentiable wave functions localized around the singular point, with an extension well beyond the singularity. Furthermore, a multitude of scattering states appear that may affect the transport and optical properties of the system.
\end{abstract}

 % Use showkeys class option if keyword display desired

\maketitle

%
%%%%%%%%%%%%%%%%%%%%%%%%%%%%%%%%%%%%%%%%%%%%%%%%%%%%%%%%%%%%%%%%%%%%%%%%%%%%%%%%%%%%%%%%%%%%%%%%%%%%%%%%%%%%%%%%%%%%%%%%%%%%%%%%%%%
%%%%%%%%%%%%%%%%%%%%%%%%%%%%%%%%%%%%%%%%%%%%%%%%%%%%%%%%%%%%%%%%%%%%%%%%%%%%%%%%%%%%%%%%%%%%%%%%%%%%%%%%%%%%%%%%%%%%%%%%%%%%%%%%%%%
% Section I
\section{Introduction}
\label{sec:Introduction}
In a quantum wire particles are confined to effectively move in one dimension while being confined in the others. Known examples are charge carriers in carbon nanotubes, semiconductor-based nanowires, and high-mobility quantum wires \citep{deshpande2008one, bockrath1999luttinger, schafer2008new, 10.1063/1.4966546, 10.1063/5.0198590, PhysRevLett.77.4612}. Transport in a chain of $\pi$-conjugated molecules \citep{doi:10.1021/acs.accounts.7b00493} can also be viewed as effectively onedimensional and is considered a key element in molecular electronics \citep{doi:10.1021/acs.chemrev.5b00680, Aradhya2013}. Of immediate relevance to this work are bent wires which can be realized, e.\,g., via mechanically strained wires (also called mechanically controlled break junctions). In this setting the wires are mechanically stressed  as to deform, and in some cases the wires may develop sections with a singular curvature \citep{mcbj, doi:10.1142/7434, doi:10.1126/science.1060294, https://doi.org/10.1002/anie.200352179}. Models have been developed to described the current-voltage characteristics in a transport experiment which is probably the most studied case.
\par
Here, we are concerned with the development of a model to describe quantum wires with strong local deformations, which in the extreme case result in a singular behavior. On a more fundamental theory level, which is the focus of this paper, there is a difference between the descriptions of particles that move freely  along  one-dimensional curve and of particles that move along a curve that is geometrically embedded in   higher-dimensional hyperspace. The former case is dealt with in a standard way by energetically confining the motion, for example by spatially dependent scalar potential that can be realized  via  appropriately applied spatially dependent voltages. The latter case of geometric embedding in a hyperspace is 
more delicate. In fact the embedding in a hyperspace may result in so-called geometry-induced effects such as the formation of bound states which should influence the transmittance and may provide a local energy gradient that attracts nearby objects.
%Here we are concerned with the possible formation of geometry-induced bound states due to locally strong deformations which results in the extreme case in a singular behavior of the quantum wire. Such states are expected to influence the transmittance and may provide a local energy gradient that attracts nearby objects. On a more fundamental theory level, which is the focus of this paper, there is a difference whether the particles are free to move along the curve, e.g., due to appropriately applied spatially dependent voltages, or whether the line of motion is embedded in higher-dimensional (hyper)space. The former case is dealt with in a standard way where the voltage enters the equation of motion as a spatially dependent scalar potential.
\par
%In the realm of nanotechnology there exists the field of investigating molecular structures, including the study of nanowires. For example, one way of constructing them is to utilize (carbon) nanotubes as channels for the drift of ions. Mathematically, situations such as this one mean the consideration of quantum mechanics under spatial constraints. 
A particular theory approach to the hyperspace situation is the confinement potential approach (CPA) \citep{Jensen1971, Costa1981, Costa1982, Maraner1995, Schuster2003, Maraner2008, Brandt2017}. In brief, one constructs an adapted Schrödinger equation to describe a particle whose motion is confined to  an isometrically embedded Riemannian submanifold $(\mathcal{M}, g)$ of (a higher-dimensional) Euclidean space from which the metric and connection are inherited. This case amounts to a dimensional reduction. Hallmarks of the hyperspace influence on the particle motion along a deformed line are the appearances of geometry-induced effects such as bound states, or gauge and potential-like fields.
\par
Generally, the CPA was considered (also in a relativistic extension) for spin-0 particles (Schrödinger and Klein-Gordon equations) \citep{Jensen1971, Costa1981, Maraner1995, Brandt2017, Ferrari2008} and spin-1/2 particles (Pauli and Dirac equations) \citep{Burgess1993, Brandt2016, SanchezMonroy2024} in the presence of electromagnetic or fields with spin-orbit coupling, and for different dimensions and co-dimensions of the target manifold \citep{Takagi1992, Maraner1993, Schuster2003, Maraner2008, Gravesen2018}. Usually, the quantum operators in the dimensionally-reduced theory
%working with the solutions of the resulting Schrödinger equation 
display correction terms which depend on geometric quantities of its underlying configuration space \citep{Liu2011, Brandt2015, Wang2017}. These so-called geometry-induced contributions  have been shown to determine scattering processes \citep{Oflaz2018, Meschede2023} and influence transport \citep{Ono2009, Ono2010, Cao2019, Serafim2021, Schwager2024, Schwager2026}, or induce topological phases \citep{Pandey2016}. Also, analog systems like photonic \citep{Longhi2007, Szameit2010, Schultheiss2010} or plasmonic \citep{Valle2010} wave guides as well as spin waves \citep{10.1063/5.0048891} and superconductors \citep{BOGUSH2025109736} exhibit such a behavior when  their equations of motion become mappable onto  a Schrödinger-like wave equation.
\par
However, as the construction of the CPA requires an isometrically embedded Riemannian submanifold of Euclidean space, it sets regularity conditions on both $\mathcal{M} \subset \mathbb{R}^{n}$ as well as the metric tensor field $g$ which it is equipped with. Thus, there are geometric limitations on the range of describable configuration spaces. Among the spaces that pose a problem  because the tangent and normal vector fields as well as derived geometric quantities such as curvature exhibit anomalies, are  spaces containing sharp bends, certain  wedges, vertices, tips, or cones as well as spaces containing contact points with themselves or self-intersections. While the cone case has been considered \citep{Filgueiras2008, Filgueiras2012, Poux2014, Kimouche2022}, the other cases either were never addressed or were considered with different approaches that did not focus on the pathologic behavior \citep{Furtado2023}, or those points were excluded from the space \citep{Pitelli2024}.
\par
The focus of this article is  on plane curves containing sharp bends, specifically accounting for diverging curvature. These can be regarded as a limiting case model for strongly bent one-dimensional wires. We present a mechanism to approximate the pathological case, where the quantum equation is ill-defined, by a family of convergent regular curves and justify that their limit can be declared a valid solution to the original problem. The necessary mathematical statements concerning the convergence are proved using operator theoretic methods. This regularization ansatz relies on singular Sturm-Liouville theory and builds on results obtained by \citep{A.M.Savchuk1999, Savchuk2003}. We belief that the scheme delivers a rigorous extension of the CPA in the sense that it enables the computation of eigenfunctions and eigenvalues of the Hamilton operator that describes particles confined to wires that contain a certain class of degeneracies.
\par
The article is organized as follows. Sec. \ref{sec:problem_analysis}  analyzes the general boundaries of the CPA and classifies the degenerate cases. 
Sec.~ \ref{sec:Model_Summary}  presents  the construction of a solution for a particular subcase, namely a degenerate plane curve with diverging curvature using a regularization method. This part is backed by the appendix, including detailed descriptions and some proofs. An example is studied in Sec.~\ref{sec:Example} where the existence of bound states induced by sharp bends is evidenced. Sec.~\ref{sec:Conclusion} summarizes the findings.
%
%%%%%%%%%%%%%%%%%%%%%%%%%%%%%%%%%%%%%%%%%%%%%%%%%%%%%%%%%%%%%%%%%%%%%%%%%%%%%%%%%%%%%%%%%%%%%%%%%%%%%%%%%%%%%%%%%%%%%%%%%%%%%%%%%%%
%%%%%%%%%%%%%%%%%%%%%%%%%%%%%%%%%%%%%%%%%%%%%%%%%%%%%%%%%%%%%%%%%%%%%%%%%%%%%%%%%%%%%%%%%%%%%%%%%%%%%%%%%%%%%%%%%%%%%%%%%%%%%%%%%%%
% Section II
%\begin{widetext}
\section{The Geometric Setting of the Confinement Potential Approach: Regular vs. Degenerate Cases}
\label{sec:problem_analysis}
The formalism of the CPA \citep{Jensen1971, Costa1981, Costa1982, Maraner1995, Schuster2003, Maraner2008, Brandt2017} constructs an effective Schrödinger equation on a configuration space of reduced dimension. To be carried out in its regular form when applied to a nonrelativistic scalar particle, it demands that the configuration space is an isometrically embedded Riemannian submanifold of Euclidean space, carrying a $C^{r}$-structure where $r\in\mathbb{N}\cup\{ \infty \}$ is sufficiently high such that derived geometric quantities such as curvature or torsion appearing in the respective context are well-defined. Often, $r \geq 2$ or $r\geq 3$ is required. For the differential geometric backgrounds we refer to \citep{Kolar1993, Carmo2013, Misner2017, Jost2017, Dajczer2019} and references therein.
\par
We work here with parametrization according to the following definition.
\begin{widetext}
\par
\begin{definition}[Parametrized differentiable (smooth) submanifold of Euclidean Space]
\label{def:Submanifold}
    A subset $\mathcal{M} \subset \mathbb{R}^{n}$ is called a \textit{differentiable (smooth)} $C^{r}$-$m$-\textit{manifold}, with differentiability class $r\in\mathbb{N}\cup\{ \infty \}$, if for all points $\textbf{x}_{0}\in\mathcal{M}$ there are a neighborhood $U_{\textbf{x}_{0}}\subset\mathbb{R}^{n}$ and a mapping $\mathcal{X}_{\textbf{x}_{0}}: \mathbb{R}^{m} \supset Q_{\textbf{x}_{0}} \rightarrow U_{\textbf{x}_{0}}\cap \mathcal{M} \subset \mathbb{R}^{n}$, $\textbf{q} \mapsto \mathcal{X}_{\textbf{x}_{0}}(\textbf{q}) \coloneqq \textbf{x}$ such that
    \begin{itemize}
        \item $\mathcal{X}_{\textbf{x}_{0}}\in C^{r}\left( \mathbb{R}^{m}, \mathbb{R}^{n} \right)$,
        \item $\mathcal{X}_{\textbf{x}_{0}}$ is a homeomorphism onto its image (w.\,r.\,t. the subspace topology induced by $(\mathbb{R}^{n}, \langle \cdot, \cdot \rangle_{2})$),
        \item $\mathrm{rank}\left( \mathrm{d}\mathcal{X}_{\textbf{x}_{0}}(\textbf{q}) \right) = m\ \ \forall\,\textbf{q}\in Q_{\textbf{x}_{0}}$.
    \end{itemize}
\end{definition}
\end{widetext}
\par
Let $g$ and $\nabla$ denote the metric and Levi-Civita connection that are canonically induced as the pull-backs of the Euclidean ones by the embedding $\iota: \mathcal{M} \hookrightarrow \mathbb{R}^{n}$, then the $C^{r}$-$m$-manifold $(\mathcal{M}, g) \subset (\mathbb{R}^{n}, \delta)$ is an isometrically (by construction) embedded Riemannian submanifold of Euclidean space with $C^{r}$-structure, because of the additional requirements on the parametrization fields. It is understood from the physical point of view that the subset $\mathcal{M}\subset \mathbb{R}^{n}$ generally is parametrized in the most well-behaved way possible, to avoid artificial pathologies. This includes utilizing a maximal atlas of the highest possible differentiability class and tailoring parameter spaces to exclude multiple coverings.
\par
The directional parts of those tangent vector fields that form the canonical coordinate frame field $\left\{ \textbf{t}_{\mu} \right\}_{\mu=1}^{m}$ are given by the columns of the Jacobi matrix of the parametrizations, $J_{\mathcal{X}_{\textbf{x}_{0}}} = \mathrm{d}\mathcal{X}_{\textbf{x}_{0}}$. Constructing the normal spaces and choosing an orthonormal normal frame field $\{ \textbf{n}_{i} \}_{i=m+1}^{n-m}$, we can parametrize a sufficiently small neighborhood of $\mathcal{M}$ by
\begin{align}
    \mathcal{Y}_{\textbf{x}_{0}}(\textbf{q}, \textbf{y}) = \mathcal{X}_{\textbf{x}_{0}}(\textbf{q}) + y^{i}\,\textbf{n}_{i}(\textbf{q})\ .
\end{align}
From this parametrization of space, an effective dimensionally reduced equation of motion can be derived by applying the CPA formalism. The details of the computation are available in the literature \citep{Jensen1971, Costa1981, Costa1982, Maraner1995, Schuster2003, Maraner2008, Brandt2017}. In summary, the restriction of the quantum motion to $\mathcal{M}$ is provided by the confining potential $V_{\mathrm{c}}$ which has to fulfill the following conditions: (i) it must be dependent on the normal displacement variables $\textbf{y}$ only, (ii) it has to possess a deep minimum on $\mathcal{M}$, i.\,e. around $\textbf{y} = \textbf{0}$, and (iii) it has to respect the gauge group of the dimensionally reduced theory, which is necessarily a subgroup of the isometry group of the higher-dimensional space. A canonical choice for this is the harmonic potential \citep{Schuster2003, Brandt2017}. Then, the CPA program requests dividing the tangent (extrinsic) and normal (intrinsic) coordinate degrees of freedom inside the Schrödinger equation and performing a similarity transformation to rescale the wave function appropriately, such that a probability density on $\left(  \mathcal{M}, g\right)$ is well-defined. After a perturbative expansion treatment, an effective extrinsic Schrödiger equation can be distilled, which describes the particle motion on the manifold $\left( \mathcal{M}, g\right)$.
\par
There are other subsets $\mathcal{M} \subset \mathbb{R}^{n}$, though, which do not fulfill the above definition of an isometrically embedded Riemannian $C^{r}$-$m$-submanifold of Euclidean space. But nevertheless, they could appear in physics considering (strongly) bent nanowires \citep{Sprung1992, Pitelli2024, Furtado2023}, conically shaped structures \citep{Filgueiras2008, Filgueiras2012, Poux2014, Kimouche2022}, Y-junctions \citep{Bibikov2008} or quantum graphs \citep{Kurasov2024}. Precisely, these are sets $\mathcal{M}$ containing a nonempty defect subset of points near which $\mathcal{M}$ fails to fulfill all of the above described characteristics. But when any of these is violated, the original CPA formalism breaks down in the sense that it cannot be carried out regularly anymore. To treat these situations, too, the formalism of the CPA needs to be extended. Often, generalized concepts are applied while solution strategies involve dissections of the prospective configuration space $\mathcal{M}$ and regularization theory. Some special cases in low dimensions have already been treated in the literature (see the remark below).
\par
In what follows, we will give an overview of different problems occurring for parametrized subsets and then provide a formal definition of the degenerate case.
\begin{widetext}
\begin{enumerate}[label=(\roman*)]
    \item The differentiable structure carried by $(\mathcal{M}, g)$ is not sufficiently regular.\\[3pt]
    The application of the CPA demands that $(\mathcal{M}, g)$ carries a $C^{r}$-structure for some $r\in \mathbb{N}: r\geq 2$, but the parametrization fields in the atlas only generate a $C^{k}$-structure with $k\in\mathbb{N}: k \leq r-1$. Define the defect set as
    \begin{align*}
        \forall\,r\in\mathbb{N}&: r\geq 2:\\
        &\hspace{0.5cm} \Xi_{C^{r}}(\mathcal{M}) \coloneqq \left\{ \textbf{x} \in \mathcal{X}_{\textbf{x}_{0}}(Q_{\textbf{x}_{0}}) \subset  \mathcal{M}:\ \mathcal{X}_{\textbf{x}_{0}}\in C^{k}(\mathbb{R}^{m}, \mathbb{R}^{n})\ \vert\ k\in\mathbb{N}: k \leq r-1 \ \text{maximally} \right\}\ .
    \end{align*}
    \item $\mathcal{M}$ does not carry a differentiable structure.\\[3pt]
    This is the limiting case of (i) where the parametrization fields in the atlas are not differentiable so that they yield a $C^{0}$-structure only and fail to generate a differentiable manifold. Consequently, tangent spaces and thus Riemannian metrics as well as normal spaces cannot be defined. Define the defect set as
    \begin{align*}
       \Xi_{C^{1}}(\mathcal{M}) \coloneqq \left\{ \textbf{x} \in \mathcal{X}_{\textbf{x}_{0}}(Q_{\textbf{x}_{0}}) \subset  \mathcal{M}:\ \mathcal{X}_{\textbf{x}_{0}}\in C^{0}(\mathbb{R}^{m}, \mathbb{R}^{n})\ \text{maximally} \right\}\ .
    \end{align*}
    \item $\hat{\mathcal{X}}_{\textbf{x}_{0}}^{\ast}\iota$ is an injective immersion, but not an embedding.\\[3pt]
    The subspace topology of Euclidean space $\mathcal{T}_{\mathbb{R}^{n}\downarrow \mathcal{M}}$ does not coincide with the image topology $\mathcal{X}_{\textbf{x}_{0}}(\mathcal{T}_{\mathbb{R}^{m}})$ of $\mathcal{M}$, locally, so that $\mathcal{M}$ can accumulate onto itself. Then, a normal neighborhood of $\mathcal{M}$ cannot be defined. Define the defect set as
    \begin{align*}
        \Xi_{\mathrm{emb.}}(\mathcal{M}) \coloneqq \left\{ \textbf{x} \in \mathcal{X}_{\textbf{x}_{0}}(Q_{\textbf{x}_{0}}) \subset  \mathcal{M}:\ \mathcal{X}_{\textbf{x}_{0}}\ \text{is a local homeomorphism w.\,r.\,t.}\ \mathcal{X}_{\textbf{x}_{0}}(\mathcal{T}_{\mathbb{R}^{m}}) = \mathcal{T}_{\mathcal{M}}\ \text{but not w.\,r.\,t.}\ \mathcal{T}_{\mathbb{R}^{n}\downarrow \mathcal{M}} \right\}\ .
    \end{align*}
    \item $\hat{\mathcal{X}}_{\textbf{x}_{0}}^{\ast}\iota$ is an immersion, but not an injective immersion.\\[3pt]
    This is the first of two limiting cases of (iii). $\hat{\mathcal{X}}_{\textbf{x}_{0}}^{\ast}\iota$ is not a homeomorphism anymore, locally, due to the loss of injectivity. Then, $\mathcal{M}$ can intersect itself, the tangent and normal spaces are not unique and a normal neighborhood of $\mathcal{M}$ cannot be defined. Define the defect set as
    \begin{align*}
        \Xi_{\mathrm{inj.\,imm.}}(\mathcal{M}) \coloneqq \left\{ \textbf{x} \in \mathcal{X}_{\textbf{x}_{0}}(Q_{\textbf{x}_{0}}) \subset  \mathcal{M}:\ \exists\, \textbf{q}_{1},\textbf{q}_{2} \in Q_{\textbf{x}_{0}}: \mathcal{X}_{\textbf{x}_{0}}(\textbf{q}_{1}) = \mathcal{X}_{\textbf{x}_{0}}(\textbf{q}_{2}) \right\}\ .
    \end{align*}
    \item $\hat{\mathcal{X}}_{\textbf{x}_{0}}^{\ast}\iota$ is not an immersion; the tangent space of $\mathcal{M}$ locally collapses so that $g$ degenerates.\\[3pt]
    This is the second of two limiting cases of (iii). The parametrization is injective and differentiable, but a canonical coordinate vector field in the frame vanishes so that the local tangent space collapses by a fall of dimension. This is accompanied by a zero in the tensor field $g$, so it degenerates and fails to be a (Riemannian) metric. Additionally, the tangent mapping of $\hat{\mathcal{X}}_{\textbf{x}_{0}}^{\ast}\iota$ is not injective so that $\iota$ fails to be an immersion, and the local normal space cannot be uniquely defined. Define the defect set as
    \begin{align*}
        \Xi_{\mathrm{imm.}}(\mathcal{M}) \coloneqq \left\{ \textbf{x} \in \mathcal{X}_{\textbf{x}_{0}}(Q_{\textbf{x}_{0}}) \subset  \mathcal{M}:\ \mathrm{rank}(\mathrm{d}\mathcal{X}_{\textbf{x}_{0}}(\textbf{q})) \leq m-1\,\Leftrightarrow\,\mathrm{nullity}(\mathrm{d}\mathcal{X}_{\textbf{x}_{0}}(\textbf{q})) \geq 1 \right\}\ .
    \end{align*}
\end{enumerate}
\end{widetext}
\begin{remark}
    The above listed defects are fundamental in the sense that they cannot be removed by choosing a different parametrization.
    \par    
    -Defect class (i) contains \enquote{mild irregularities} in the sense that the here appearing spaces $(\mathcal{M}, g)$ are still isometrically embedded Riemanninan submanifolds carrying (at least) a $C^{1}$-structure. Especially, the normal spaces are defined everywhere and the fundamental idea of the CPA, which is arranging a confinement potential in a neighborhood along the normal degrees of freedom, is not perturbed. This is in variance with all the other classes. However, the regularity of the structure is not strong enough to guarantee well-behaved (continuous) derived quantities. A special case of this defect class are manifolds where these quantities have the value of a regular distribution, which then can be captured by weak differentiation and the theory of weak solutions of the CPA-Schrödinger equation. Examples in the context of curves contain jumping potentials like the one appearing in the junction of a straight line  and a circle with $\kappa \propto H(s)$ \citep{Sprung1992} or the more general example of curves possessing a curvature such that $\kappa \in L^{1}(J)$ and $\kappa^{2}\in H^{-1}(J)$ hold, that we treat in this work.
    \par
    -Defect class (ii) contains irregularities where the tangent cannot be (uniquely) defined such as edges, vertices, or conical apices. Thus, it is a subclass of (v) in the sense that the tangent map is not injective because it does not even exist. Compared to (i), these cases are \enquote{more severe}. Derived quantities such as curvature often appear as a singular distributions like in the one-dimensional case of a vertex of a curve \citep{Pitelli2024}, where $\kappa(s) \propto \delta(s - s_{0})$ so that the geometric potential contains a product of distributions, $V_{\mathrm{geo}} \propto \delta^{2}(s-s_{0})$, which is ill-defined. In the two-dimensional case, however, conical spaces could be treated in \citep{Filgueiras2008, Filgueiras2012, Poux2014, Kimouche2022}, exploiting the rotational symmetry. Note that $\mathcal{M}$ still obtains a topological structure in this case, anyway.
    \par
    -Defect class (iii) contains irregularities where neighborhoods of the parameter $\textbf{q}\in Q_{\textbf{x}_{0}}$ are mapped to neighborhoods of $\iota(\mathcal{X}_{\textbf{x}_{0}}(\textbf{q})) \in \mathcal{M}$ that are open w.\,r.\,t. the image topology, but not the subspace topology, because the (corestricted) parametrization is not a topological homeomorphism. This allows $\mathcal{M}$ to come arbitrarily close to itself, locally accumulating onto itself.
    \par
    -Defect class (iv) contains irregularites where $\mathcal{M}$ intersects itself because $\iota(\mathcal{X}_{\textbf{x}_{0}}(\textbf{q}))$ is not injective and thus no homeomorphism anymore (as it is not invertible so that topologies cannot be compared). In such points, the tangent and normal spaces are not uniquely defined, as several independent ones exist. $\mathcal{M}$ loses the structure of a submanifold.
    \par
    -Defect class (v) contains irregularities where the parametrization is differentiable but some coordinate frame tangent fields vanish locally, leading the tangent space to collapse. Thus, $\hat{\mathcal{X}}_{\textbf{x}_{0}}^{\ast}\iota$ is locally no immersion indicating a local reduction of dimension. This includes tips and transitioning points between different connected components of $\mathcal{M}$ with invariant dimensions. $\mathcal{M}$ loses the structure of a submanifold. Note that in singularity theory, there is the following classification of points for $\mathcal{X}_{\textbf{x}_{0}}\in C^{1}(\mathbb{R}^{m},\mathbb{R}^{n})$:
    \begin{align*}
        \mathrm{rank}\left( \mathrm{d}\mathcal{X}_{\textbf{x}_{0}}\right) &= \mathrm{min}\left( m, n \right) &&: &&\text{regular}\ , \\
        \mathrm{rank}\left(\mathrm{d}\mathcal{X}_{\textbf{x}_{0}} \right) &< \mathrm{min}\left( m, n \right) &&: &&\text{singular}\ , \\
        \mathrm{rank}\left(\mathrm{d}\mathcal{X}_{\textbf{x}_{0}} \right) &< n &&: &&\text{critical}\ .
    \end{align*}
    In the context of the CPA it is $m < n$ so that all points of $\mathcal{M}$ are critical. This degeneracy class then precisely refers to the singular points of this classification. Note further that if there is an immersion, it will automatically be isometric, so that  there is no additional failure of isometry. Admitting these zeros, the tensor field $g$ degenerates and fails to be a (Riemannian) metric. However, it can still be regarded as an almost Riemannian metric (a special class of sub-Riemannian metrics), so that concepts of this and related research fields should naturally apply to these spaces and the Schrödinger equation can be interpreted accordingly \citep{Gromov2007, Boscain2013, Beschastnyi2021}. An example of this class, the Beltrami surface, was considered in \citep{Furtado2023} but this work did not address the peculiarities of the degenerate manifold.
\end{remark}
\par
We thus define the degenerate submanifold in the viewpoint of the CPA as follows.
\begin{widetext}
\par
\begin{definition}[Degenerate submanifold of Euclidean space (CPA)]
\label{def:degenerate_submanifold}
   Let $\mathcal{M}\subset \mathbb{R}^{n}$ be a parametrized object such that for all points $\textbf{x}_{0} \in \mathcal{M}$ there are a neighborhood $U_{\textbf{x}_{0}} \subset \mathbb{R}^{n}$ and a mapping $\mathcal{X}_{\textbf{x}_{0}}: \mathbb{R}^{m} \supset Q_{\textbf{x}_{0}} \to U_{\textbf{x}_{0}}\cap \mathcal{M} \subset \mathbb{R}^{n}$, $\textbf{q} \mapsto \mathcal{X}_{\textbf{x}_{0}}(\textbf{q}) \coloneqq \textbf{x}$, such that $\mathcal{X}_{\textbf{x}_{0}}\in C^{0}(\mathbb{R}^{m}, \mathbb{R}^{n})$.\\[5pt]
   For fixed $r,m\in \mathbb{N}$ the \textit{regular part} of $\mathcal{M}$ is defined as
   \begin{align*}
       \mathrm{Reg}_{r,m}(\mathcal{M}) \coloneqq &\left\{  \textbf{x}\in \mathcal{M}:\ \text{The maximal atlas}\ \left( U_{i}\cap \mathcal{M}, \mathcal{X}_{i} \right)_{i\in \mathcal{I}}\ \text{suitably covering}\  \mathcal{M}\ \text{locally produces} \right.\\
       &\hspace{1.7cm}\left. \text{the structure of a}\ C^{r}\text{-}m\text{-manifold}. \right\}\ .
   \end{align*}
   Let the defect sets be defined as above. The subset 
   \begin{align*}
       \Xi_{r,m}(\mathcal{M}) \coloneqq \mathcal{M} \setminus \mathrm{Reg}_{r,m}(\mathcal{M}) = \Xi_{C^{r}}(\mathcal{M}) \cup \Xi_{C^{1}}(\mathcal{M}) \cup \Xi_{\mathrm{emb.}}(\mathcal{M}) \cup \Xi_{\mathrm{inj.\,imm.}}(\mathcal{M}) \cup \Xi_{\mathrm{imm.}}(\mathcal{M})
   \end{align*}
   is called the \textit{degenerate locus} of $\mathcal{M}$ if it is closed and has Hausdorff measure zero, so that $\mathrm{dim}(\Xi_{r,m}(\mathcal{M})) \leq m-1$. Any space for which $\Xi_{r,m}(\mathcal{M}) \neq \emptyset$ is called a \textit{degenerate submanifold} (in the viewpoint of the CPA).
\end{definition}
\end{widetext}
\par
We remark, that the CPA needs to be extended to cope with the different degenerate submanifold classes. Except for class (i), this starts from a weakening of the CPA formalism in the sense that the usual construction of the confining potential along the normal degrees of freedom is carried out only on the (connected components of the) regular part $\mathrm{Reg}_{r,m}(\mathcal{M})$. This means it is constructed almost everywhere, i.e.~everywhere except on the degenerate locus $\Xi_{r,m}(\mathcal{M}) \subset \mathcal{M}$. We deem  this definition  useful for the description of the degenerate submanifold $\mathcal{M}$ as a stratifold \citep{Kreck2010} and the application of more abstract concepts of differential geometry, which we think will contribute to the construction of suitable CPA extensions.
\par
The above classification is not meant to be exhaustive as it only focuses on the context of the Schrödinger equation solved on configuration spaces inside Euclidean space. It is possible to consider subsets of other Riemannian manifolds with already nontrivial geometry and topology, too, where other geometric compatibility requirements possibly occur. Especially, relativistic treatments focus on different partial differential equations, Minkowski spacetime and pseudo-Riemannian geometry \citep{Brandt2017}. Consideration of spinors requires spin manifolds and spin geometry \citep{Liu2011, Brandt2015, Wang2017, Lawson1989}. So, the more mathematical demands on the configuration space the physics generally requires, the richer becomes the portfolio of irregularities. Furthermore, considering the quantum equation of motion coupled to a gauge field like the electromagnetic field introduces additional classifications of the spaces, like in the context of the Aharonov-Bohm effect, that do not interfere with the above scheme.
%\end{widetext}
%
%%%%%%%%%%%%%%%%%%%%%%%%%%%%%%%%%%%%%%%%%%%%%%%%%%%%%%%%%%%%%%%%%%%%%%%%%%%%%%%%%%%%%%%%%%%%%%%%%%%%%%%%%%%%%%%%%%%%%%%%%%%%%%%%%%%
%%%%%%%%%%%%%%%%%%%%%%%%%%%%%%%%%%%%%%%%%%%%%%%%%%%%%%%%%%%%%%%%%%%%%%%%%%%%%%%%%%%%%%%%%%%%%%%%%%%%%%%%%%%%%%%%%%%%%%%%%%%%%%%%%%%
% Section III
\section{The Regularization of Degenerate Plane Curves}
\label{sec:Model_Summary}
In the remainder of this article we focus on the case of a degenerate plane curve $\mathcal{M} \subset \mathbb{R}^{2}$ whose degenerate locus $\Xi_{2,1}(\mathcal{M}) \subset \mathcal{M}$ is then an at most countable set of discrete points (i.\,e. a zero-dimensional submanifold of $\mathcal{M}$). The regular part $\mathrm{Reg}_{2,1}(\mathcal{M}) = \mathcal{M}\setminus\Xi_{2,1}(\mathcal{M})$ is then a piecewise isometrically embedded Riemannian $C^{2}$-$1$-submanifold of Euclidean space, so that the CPA can be carried out regularly almost everywhere. After stating the result of the CPA for the regular case, we construct an extension of the CPA to solve a particular subclass of the above mentioned degenerate case. The mathematical backgrounds of the construction are rooted in \citep{A.M.Savchuk1999, Savchuk2003}. Explanations of these, definitions and propositions together with some mathematical proofs are collected in the appendix.
\subsection{The Regular Model for Plane Curves}
In the regular case, we consider a one-dimensional Riemannian submanifold $\left( \mathcal{M}, g \right)$, also called a curve, that is given as the image of the interval $Q = (a,b) \eqqcolon I \subset \mathbb{R}^{1}$ under a parametrization $\mathcal{X}_{\mathcal{M}}\in C^{2}(I, \mathbb{R}^{2})$. Then, the canonical coordinate tangent vector field $\textbf{t}$ is given by the Jacobi matrix $J_{\mathcal{M}} = \mathrm{d}\mathcal{X}_{\mathcal{M}} \doteq \textbf{t}$. As an embedded submanifold, the curve $\mathcal{M}$ is then equipped with a $C^{2}$-structure and the component of the metric tensor field fulfills $g_{11} = \Vert \textbf{t}\Vert_{2}^{2} \in C^{1}(I)$. The Levi-Civita connection is assumed to act on the tangent bundle $T\mathcal{M}$. There exists a unique parametrization of the curve with respect to the arc length
\begin{align}
    s(q) = s_{0} + \int\limits_{q_{0}}^{q} \Vert \textbf{t}(\tau) \Vert_{2}\ \mathrm{d}\tau\ ,\quad q_{0}\in \overline{I}\ .
\end{align}
Using this formula, any general parametrization can be transformed into one with respect to the arc length, with arc length interval $J\subset \mathbb{R}^{1}$. For simplicity, we will assume $q = s$ in the following. By elementary relations of differential geometry, which can be found in e.\,g. \citep{Carmo2016, Jost2017}, we can construct the Frenet-Serret frame field consisting of the tangent $\textbf{t}$ and normal $\textbf{n}$ vector fields fulfilling
\begin{align}
   \frac{\mathrm{d}}{\mathrm{d} s} \begin{pmatrix}
        \textbf{t}(s) \\
        \textbf{n}(s)
    \end{pmatrix} = \begin{pmatrix}
0 & \kappa(s) \\
-\kappa(s) & 0 
    \end{pmatrix} \cdot \begin{pmatrix}
        \textbf{t}(s) \\
        \textbf{n}(s)
    \end{pmatrix}\ ,
\label{eq:Frenet-Serret}
\end{align}
where $\kappa$ is the curvature of the curve. Note that in arc length parametrization these two vectors form an orthonormal frame field. By the fundamental theorem of the local theory of curves, it is known that for any scalar functions $\kappa\in C^{0}(I)$, with $\kappa(s) > 0\ \forall\, s\in J$, there exists a unique (up to rigid motion) parametrized curve possessing these functions as curvature and torsion, with $s$ serving as the arc length parameter.
\par
For any regular curve $\mathcal{M} \subset \mathbb{R}^{n}$ the Frenet-Serret frame field can be parametrized by an angle obtained as the integral curvature,
\begin{align}
    \gamma(s) = \gamma_{0} + \int\limits_{s_{0}}^{s} \kappa(s')\ \mathrm{d}s'\ ,
\end{align}
and the turn of the curve is defined as the total integral curvature \citep{Alexandrov1989}
\begin{align}
    \gamma(\mathcal{M}) \coloneqq \int\limits_{J} \kappa(s')\ \mathrm{d}s'\ .
\label{eq:turn_of_curve}
\end{align}
\par
Using the normal vector field, we can parametrize a sufficiently small normal neighborhood of the curve by
\begin{align}
    \mathcal{Y}(s, y) = \mathcal{X}_{\mathcal{M}}(s) + y\,\textbf{n}(s)\ .
\end{align}
Executing the CPA program for our case yields the onedimensional Schrödinger equation
\begin{align}
    \ii\hbar\,\partial_{t}\chi_{\mathrm{t}} = -\frac{\hbar^{2}}{2m}\,\left[ \partial_{s}^{2} + \frac{\kappa(s)^{2}}{4}\,\hat{I} \right]\chi_{\mathrm{t}}\ ,
\label{eq:TangentialGleichung}
\end{align}
which for a field-free particle contains a correction term in dependence of the curvature of the curve that is often denoted as the geometric potential,
\begin{align}
    V_{\mathrm{geo}}(s) \coloneqq -\frac{\hbar^{2}}{8m}\kappa^{2}(s)\ .
\end{align}
\subsection{Modeling Degenerate Curves Using Sturm-Liouville Theory}
Now, we construct an extension of the CPA formalism to incorporate a particular subclass of degenerate plane curves whose degenerate locus is associated with a curvature function $\kappa\in L^{1}(J)$ \footnote{In our notation it is $\Xi_{2,1}(\mathcal{M}) = \Xi_{C^{2}}(\mathcal{M})$ and our problem belongs to degeneracy class (i).}. In the CPA, such curvatures are generally interpreted as discontinuous or singular geometric potentials by the effective Schrödinger equation \eqref{eq:TangentialGleichung}. We will particularly focus on singular potentials and show, that if additionally $\kappa^{2}\in H^{-1}(J)$ holds, the corresponding Hamilton operator possesses a well-defined spectrum using a regularization ansatz and operator theoretic methods. Precisely, the degenerate curve is approximated by a family of regular ones, $(\mathcal{X}_{\mathcal{M}_{\varepsilon}})_{\varepsilon \in \mathbb{R}^{+}}$, for which the CPA works according to the previous section. This naturally translates into studying a convergent series of Hamilton operators. We then show that both the eigenvalues and the eigenfunctions of the regularized Schrödinger equations converge in a proper sense so that their limit provides a reasonable notion of the solution to the original degenerate problem.
\par
We consider a curve with trace $\mathcal{M}$ containing an isolated irregularity in the point $\mathbf{x}_{0}\in \Xi_{2,1}(\mathcal{M})$. Let $I = (a,b)$ be an open and bounded interval in $\R$, with $\mathcal{X}^{-1}_{\mathcal{M}}(\textbf{x}_{0}) =q_{0}\in I$. Thus, the parametrization function
\begin{equation}
\label{Eq:singularCurve}
    \function{\mathcal{X}_{\mathcal{M}}}{I}{\R^2},\ q\mapsto \mathcal{X}_{\mathcal{M}}(q)
\end{equation}
is regular on $I \setminus \lbrace q_{0} \rbrace$ and irregular in $q_{0}$. Note that the point of degeneracy can always be set as $q_{0} = 0$, so that $\textbf{x}_{0} = \mathcal{X}_{\mathcal{M}}(0)$, since the image of the curve is invariant under diffeomorphisms and corresponding redefinitions of the domain, e.\,g. translations. We observe that $I \setminus \lbrace 0 \rbrace$ is open so the arc length is well-defined there. Thus, let $J \setminus \lbrace 0 \rbrace$ be the corresponding open interval for the parametrization by arc length. The CPA construction is well-defined on $J\setminus \{ 0 \}$ and delivers an effective extrinsic Schrödinger equation in the naive pointwise sense. Then, after the separation of time in \eqref{eq:TangentialGleichung} it reads
\begin{equation}
     \label{Eq:SingularTangent}
    -\frac{\hbar^{2}}{2m}\left[ \partial_s^2+\frac{\kappa(s)^2}{4}\,\hat{I} \right]\chi_{\mathrm{t}}(s) = E\,\chi_{\mathrm{t}}(s)\ ,\; s\in J \setminus \lbrace 0 \rbrace\ ,
\end{equation}
where the curvature function is assumed as $\kappa\in L^{1}(J)$ and diverges as
\begin{align}
    \lim_{s\to 0} \vert \kappa(s)\vert = +\infty\ .
\end{align}
Note that the assumption of considering an $L^{1}(J)$-function restricts the strength of the singularity. To be precise, the problem we face is to compute the spectrum of a Schrödinger equation with a Hamilton operator containing a potential-like term which generates a distribution of first order, meaning there exists $U_{\mathrm{geo}}\in L^{2}(J)$ such that $U_{\mathrm{geo}}' = V_{\mathrm{geo}}$ in the distributional sense, i.\,e. $\langle U_{\mathrm{geo}}, h'\rangle_{L^{2}(J)} = - \langle V_{\mathrm{geo}}, h\rangle_{L^{2}(J)}$ for all test functions $h\in C_{c}^{\infty}(J)$. This is equivalent to stating $V_{\mathrm{geo}} \in H^{-1}(J)$. Note that no other assumptions are made for $V_{\mathrm{geo}}$, so it may be an actual generalized function, in general. See appendix sec.\,\ref{sec:Appendix_Convergence_of_Spectra} for a more detailed discussion of the mathematical background rooted in \citep{A.M.Savchuk1999, Savchuk2003}.
\par
To deal with the anomaly in the curvature, we now suggest the following regularization ansatz. We define a family of curves $(\mathcal{X}_{\mathcal{M}_{\varepsilon}})_{\varepsilon \in \mathbb{R}^{+}}$ that shall approximate the degenerate curve $\mathcal{X}_{\mathcal{M}}$ in the limit $\varepsilon \to 0$. This approximating family necessarily has to converge with two distinct qualities so that our treatment works:
\begin{enumerate}[label = (\roman*)]
    \item The parametrizations and all derived geometric quantities have to converge in the pointwise sense for all $s\in J\setminus \mathcal{X}_{\mathcal{M}}^{-1}[\Xi_{2,1}(\mathcal{M})]$, i.\,e.
    \begin{align}
    \begin{aligned}
        \lim_{\varepsilon\to 0}\left\vert \mathcal{X}_{\mathcal{M}_{\varepsilon}}(s) - \mathcal{X}_{\mathcal{M}}(s)\right\vert &= 0\ ,\\
        \lim_{\varepsilon\to 0}\left\vert\kappa_{\varepsilon}(s) - \kappa(s)\right\vert &= 0\ .
    \end{aligned}
    \end{align}
    \item The primitives of the curvature-induced potential functions have to converge in the $L^{2}(J)$-sense, i.\,e.
    \begin{align}
        \lim_{\varepsilon \to 0} \left\Vert {U_{\mathrm{geo}}}_{\varepsilon} - U_{\mathrm{geo}} \right\Vert_{L^{2}(J)} = 0\ .
    \end{align}
\end{enumerate}
The first statement means the convergence of curves as geometric objects. Obviously, this is a necessary condition for our solution scheme to apply. We call any regularization fulfilling (i) a \textit{pre-admissible regularization}. The second statement is necessary to secure the correct convergence of the series of Hamilton operators associated with the curves when we abstract the discussion from geometric objects to partial differential equations. We call any regularization fulfilling both (i) and (ii) an \textit{admissible regularization}. Note that the second statement cannot be dropped, and see appendix sec.\,\ref{sec:Appendix_admissible_regs} for further discussions.
\par
A sufficient condition for constructing a pre-admissible regularization is to provide a family of functions $\lbrace \kappa_\varepsilon\rbrace_{\varepsilon\in\R^{+}} \subset C^{0}(J)$ so that $\kappa_{\varepsilon}$ is bounded on $J$ for every $\varepsilon\in \mathbb{R}^{+}$ and that the family converges in $L^1(J)$,
\begin{equation}
    \lim\limits_{\varepsilon\rightarrow0}\norm{\kappa_\varepsilon - \kappa}_{L^{1}(J)} = 0\ .
\end{equation}
By the Frenet-Serret equations \eqref{eq:Frenet-Serret}, every $\kappa_{\varepsilon}$ defines a unique (up to rigid motion) regular parametrized curve in arc length parametrization, for which it serves as the curvature function. The $L^{1}(J)$-convergence of the curvatures implies pointwise convergence of the corresponding curves, their traces in the sense of geometric objects, and derived geometric quantities (see appendix sec.\,\ref{sec:Appendix_Convergence_of_Curves}).
\par
Due to the regularity of all curvature functions $\kappa_{\varepsilon}$ we consider the corresponding Schrödinger equations
\begin{align}
    \label{Eq:DesingularizedTangent}
    -\frac{\hbar^{2}}{2m}\left[ \partial_s^2+\frac{\kappa_\varepsilon(s)^2}{4}\,\hat{I} \right]{\chi_{\mathrm{t}}}_{\varepsilon}(s) = E_{\varepsilon}\,{\chi_{\mathrm{t}}}_{\varepsilon}(s)\ ,\ s \in J\ .
\end{align}
It is now the standard situation encountered in quantum mechanics that, for all $\varepsilon\in \mathbb{R}^{+}$, the self-adjoint realizations of their Hamilton operators, $(\hat{H}_{\varepsilon}, \mathcal{D}(\hat{H}_{\varepsilon}))$, are given by the domains
\begin{align}
    \mathcal{D}(\hat{H}_{\varepsilon}) = \left\{ \varphi\in H^{2}(J): \hat{U}_{1}(\varphi) = \hat{U}_{2}(\varphi) = 0 \right\} \subset H^{2}(J)
\end{align}
that are specified by some boundary linear forms $\hat{U}_{1}$ and $\hat{U}_{2}$ (see appendix sec.\,\ref{sec:Appendix_Convergence_of_Spectra} for details). This construction contains the familiar Dirichlet, Neumann, Robin, and periodic boundary conditions. Consequently, for every fixed $\varepsilon \in \mathbb{R}^{+}$ the spectrum is discrete and there exists a well-defined orthonormal basis of the Hilbert space $(L^{2}(J), \langle \cdot, \cdot\rangle_{L^{2}(J)})$ composed of nontrivial solutions $\left\{ ({\chi_{\mathrm{t}}^{n}}_{\varepsilon}, E_{\varepsilon}^{n})\right\}_{n \in \mathbb{N}}$ (counted without multiplicity) to this equation. Moreover, each spectrum is semi-bounded from below, such that for fixed $\varepsilon \in \mathbb{R}^{+}$
\begin{align}
    \inf \left\{ E_{\varepsilon}^{n} \right\}_{n \in \mathbb{N}} >  -\frac{\hbar^2}{8m}\,\sup_{s\in J}\kappa_\varepsilon^{2}(s)  > -\infty\ ,
\end{align}
since the geometric potential ${V_{\mathrm{geo}}}_{\varepsilon}$ is bounded.
\par
Besides, now considering the singular case, we can reformulate the Schrödinger equation \eqref{Eq:SingularTangent} by introducing the (adapted) quasi-derivative
\begin{align}
    \chi_{\mathrm{t}}^{[1]} \coloneqq \partial_{s}\chi_{\mathrm{t}} - u_{\mathrm{geo}}\,\chi_{\mathrm{t}}\ ,\quad u_{\mathrm{geo}} \coloneqq \frac{2m}{\hbar^{2}}\,U_{\mathrm{geo}}
\end{align}
as
\begin{align}
    \frac{\hbar^{2}}{2m}\,\hat{L}\,\chi_{\mathrm{t}} = E\,\chi_{\mathrm{t}}\ ,
\end{align}
with the linear operator action
\begin{align}
    \hat{L} = \frac{2m}{\hbar^{2}}\,\hat{H} = -\left( \chi_{\mathrm{t}}^{[1]}\right)' - u_{\mathrm{geo}}\,\chi_{\mathrm{t}}^{[1]} - u_{\mathrm{geo}}^{2}\,\chi_{\mathrm{t}}\ .
\end{align}
$\hat{L}$ does not contain a singular potential anymore, and using the adapted quasi-derivative it was shown in \citep{A.M.Savchuk1999, Savchuk2003} that many results of regular Sturm-Liouville problems also hold in this singular case, i.\,e. the special case where the divergent term fulfills $V_{\mathrm{geo}}\in H^{-1}(J)$. Especially, it was shown that any self-adjoint realization of the singular Hamilton operator, $(\hat{H}, \mathcal{D}(\hat{H}))$, is given by the domain
\begin{align}
\begin{aligned}
    \mathcal{D}(\hat{H}) = &\left\{ \varphi\in L^{2}(J): \right.\\
    & \hspace{0.5cm}\left. \varphi, \varphi^{[1]} \in W^{1,1}(J)\ ,\ \hat{H}\varphi\in L^{2}(J)\ ,\right.\\
    &\hspace{0.5cm} \left. \hat{\tilde{U}}_{1}[\varphi] = \hat{\tilde{U}}_{2}[\varphi] = 0 \right\} \subset H^{1}(J)\ ,
\end{aligned}
\end{align}
where $\hat{\tilde{U}}_{1}[\varphi]$ and $\hat{\tilde{U}}_{2}[\varphi]$ denote adapted boundary linear forms similar to those above. Note that the solutions of the singular equation thus have less structure, in general, as they do not necessarily possess a second (weak) derivative. The eigenstates solving the Schrödinger equation, $(\chi_{\mathrm{t}}^{n}, E^{n})$, counted again without multiplicity, also form an orthonormal basis of the Hilbert space $(L^{2}(J), \langle \cdot, \cdot \rangle_{L^{2}(J)})$.
\par
Then, according to \citep{A.M.Savchuk1999, Savchuk2003}, if
\begin{align}
    \norm{ {U_{\mathrm{geo}}}_{\varepsilon} - U_{\mathrm{geo}} }_{L^{2}(J)} \to 0\ ,\quad (\varepsilon \to 0)\ ,
\end{align}
i.\,e. if $(\mathcal{X}_{\mathcal{M}_{\varepsilon}})_{\varepsilon \in \mathbb{R}^{+}}$ is an admissible regularization, the spectra $\{\sigma(\hat{H}_{\varepsilon})\}_{\varepsilon\in \mathbb{R}^{+}}$ converge towards the spectrum of the degenerate case, $\sigma(\hat{H})$, which is real, discrete, and semi-bounded from below, and for all $n\in\mathbb{N}$ the eigenvalues fulfill
\begin{align}
    \vert E_{\varepsilon}^{l_{n}} - E^{n}\vert_{\mathbb{R}} \to 0\ ,\quad (\varepsilon \to 0)\ ,
\end{align}
where the variables $l_{n}$ identify those eigenstates that become degenerate in the limit, which is possible if there occurs a breaking of symmetry during the regularization (for one-dimensional Sturm-Liouville equations we can produce a two-fold degeneracy at the highest). This provides the convergence of eigenvalues and the limit gives the solution to the degenerate case. Furthermore, as we show in appendix sec.\,\ref{sec:Appendix_Convergence_of_Curves}, the associated wave function of the $n$-th eigenstate of the limiting operator can also be obtained by observing the convergence of ${\chi_{\mathrm{t}}^{l_{n}}}_{\varepsilon}$ for $\varepsilon \to 0$.
\par
Thus, the particle confined to a degenerate curve (with finite length) possesses a well-defined set of quantum states, and these can be computed as limits using admissible regularizations. Analogous statements can be made for infinite curves with $J=\R$ but then the spectra will contain a continuous part. However, our proofs do not treat this case.
%
%%%%%%%%%%%%%%%%%%%%%%%%%%%%%%%%%%%%%%%%%%%%%%%%%%%%%%%%%%%%%%%%%%%%%%%%%%%%%%%%%%%%%%%%%%%%%%%%%%%%%%%%%%%%%%%%%%%%%%%%%%%%%%%%%%%
%%%%%%%%%%%%%%%%%%%%%%%%%%%%%%%%%%%%%%%%%%%%%%%%%%%%%%%%%%%%%%%%%%%%%%%%%%%%%%%%%%%%%%%%%%%%%%%%%%%%%%%%%%%%%%%%%%%%%%%%%%%%%%%%%%%
% Section IV
\section{Example}
\label{sec:Example}
For illustration  we consider a plane degenerate curve which is embedded in two-dimensional space. Precisely, the singular curvature function $\kappa \in L^{1}(J)$ shall be \footnote{$\kappa \in L^{1}(J)$ for $\alpha \in (0,1)$ but the requirement $\kappa^{2}\in H^{-1}(J)$ puts additional restrictions.}
\begin{align}
    \kappa(s) = K\,(s^{2})^{-\frac{\alpha}{2}} = K\,\vert s\vert ^{-\alpha}\ ,\quad \alpha\in \left(0,\, \frac{3}{4}\right)\ .
\end{align}
This example belongs to the degeneracy class (i), i.\,e. it is an isometrically embedded $C^{1}$-curve whose curvature diverges in a single point, that we put as the origin of our coordinate system, $\Xi_{2,1}(\mathcal{M}) = \{ \textbf{0} \}$. Note that we chose a convention where we measure the arc length from this point on, but along both directions individually, so that the arc length parameter $s$ is allowed to obtain negative values. The parameter $K$ is related to the turn of the curve \eqref{eq:turn_of_curve} via the relation $\gamma(\mathcal{M}) = \pi - \theta$ as 
\begin{align}
    K = \frac{(1- \alpha)\,(\pi - \theta)}{\vert a \vert^{1-\alpha} + \vert b\vert^{1 - \alpha}}\ ,
\end{align}
where $\theta$ is the opening angle between the tangents on the endpoints.
\par
Specifically, we consider a particle with mass \linebreak$m = m_{e}$ on the degenerate curve with arc length interval \linebreak$J = (-5,5)\,\mathrm{nm}$ and geometric parameters $\alpha = \frac{1}{2}$ and $\theta = \frac{\pi}{8}$. Then, the geometric potential $V_{\mathrm{geo}}\in H^{-1}(J)$ and its primitive $U_{\mathrm{geo}} \in L^{2}(J)$ fulfills with $C\in\mathbb{R}$
\begin{align}
\begin{aligned}
    V_{\mathrm{geo}} &= -\frac{\hbar^{2}\,K^{2}}{8m}\,\vert s\vert^{-1}\ ,\\
    U_{\mathrm{geo}} &= - \frac{h^{2}\,K^{2}}{8m}\,\mathrm{sgn}(s)\,\mathrm{ln}(\vert s\vert) + C\ .
\end{aligned}
\end{align}
This means, the limiting case $\alpha = \frac{1}{2}$  defines a curve whose geometric potential resembles an attractive one-dimensional Coulomb potential centered at the degenerate locus. This has an amplitude strength that increases when the curve length or the particle mass decrease, or the turn increases (the opening angle decreases). Note that $V_{\mathrm{geo}} \notin L^{1}(J)$ as it contains a non-integrable singularity at $s=0$.
\par
As a regularization, we apply the curves $\{ \mathcal{X}_{\mathcal{M}_{\varepsilon}} \}_{\varepsilon\in\mathbb{R}^{+}}$ with curvatures $\{\kappa_{\varepsilon}\}_{\varepsilon \in \mathbb{R}^{+}} \subset C^{0}(J)$, given as
\begin{align}
    \kappa_{\varepsilon} = K\,(\vert s\vert + \varepsilon)^{-\alpha}\ ,
\end{align}
which for $\alpha = \frac{1}{2}$ implies
\begin{align}
    \begin{aligned}
    {V_{\mathrm{geo}}}_{\varepsilon} &= -\frac{\hbar^{2}\,K^{2}}{8m}\,\left( \vert s\vert + \varepsilon \right)^{-1}\ ,\\
    {U_{\mathrm{geo}}}_{\varepsilon} &= - \frac{h^{2}\,K^{2}}{8m}\,\mathrm{sgn}(s)\,\mathrm{ln}\left(\vert s\vert + \varepsilon \right) + C_{\varepsilon}\ .
\end{aligned}
\end{align}
It is easily verified that this is an admissible regularization, requiring $C_{\varepsilon} \to C$, $(\varepsilon\to 0)$. Some members of this family of regularized curves and the limiting degenerate curve are depicted in figure \ref{fig:ExampleCurves}. For $\varepsilon\to 0$, they converge in the pointwise sense as subsets of the Euclidean space. Similarly, figure \ref{fig:ExampleCurvature} shows some regularized curvature functions and the singular curvature in dependence of the arc length parameter. Note  that regularized curvatures and thus the corresponding geometric potentials are bounded. The family of curvatures approximates the limiting case in the $L^{2}(J)$-sense, when $\varepsilon\to 0$.
\begin{figure}
    \centering
    \includegraphics[width=\linewidth]{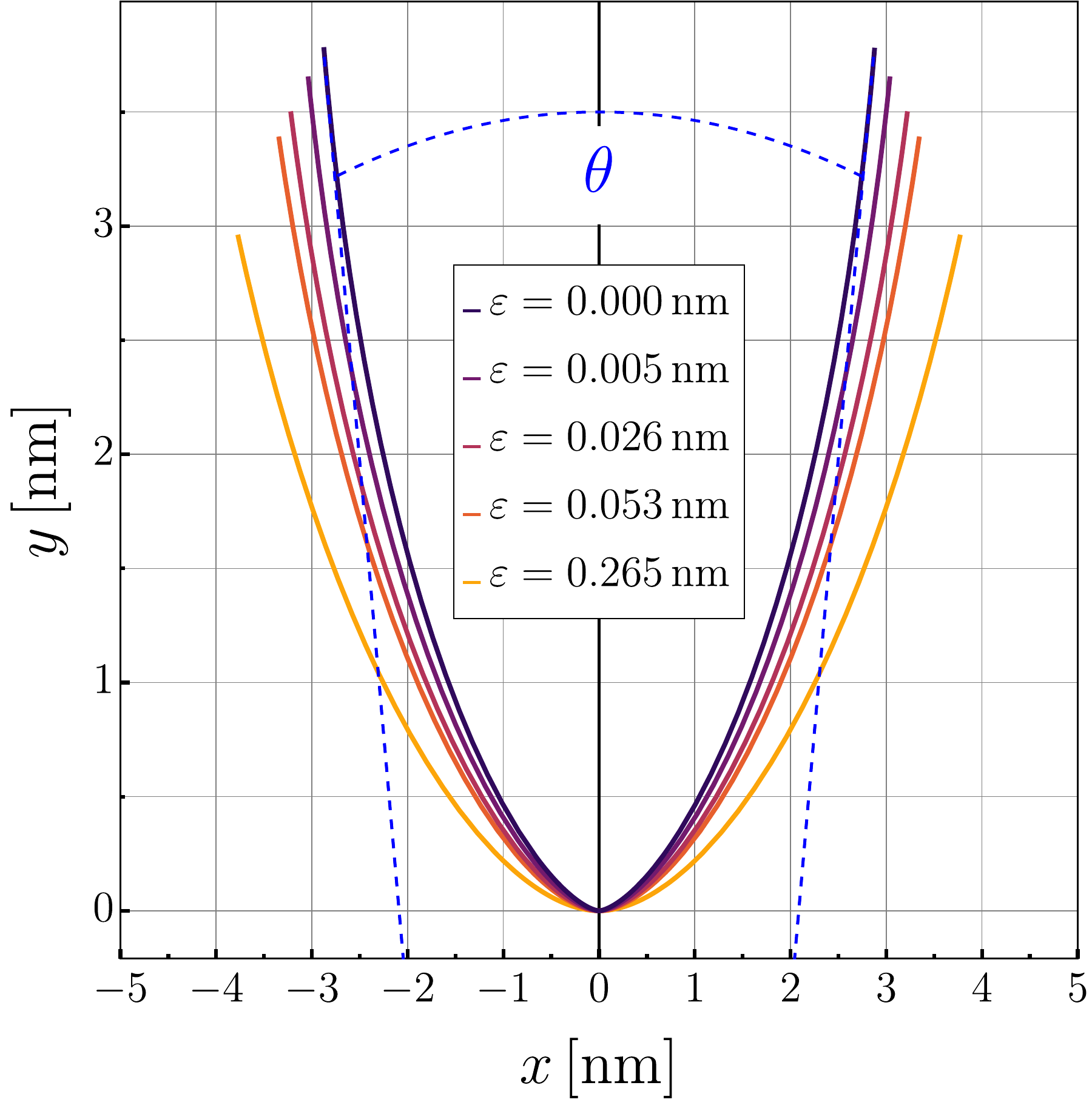}
    \caption{Four representatives of the curve family $\{\mathcal{X}_{\mathcal{M}_{\varepsilon}}\}_{\varepsilon \in \mathbb{R}^{+}}$ and the limiting degenerate curve $\mathcal{X}_{\mathcal{M}}$. $\theta$ is the opening angle of the latter. The subsets of Euclidean space converge in the pointwise sense. Parameters: $J = (-5, 5)\,\mathrm{nm}$, $\alpha = \frac{1}{2}$, $\theta = \frac{\pi}{8}$, $m = m_{e}$}
    \label{fig:ExampleCurves}
\end{figure}
\begin{figure}[h]
    \centering
    \includegraphics[width=\linewidth]{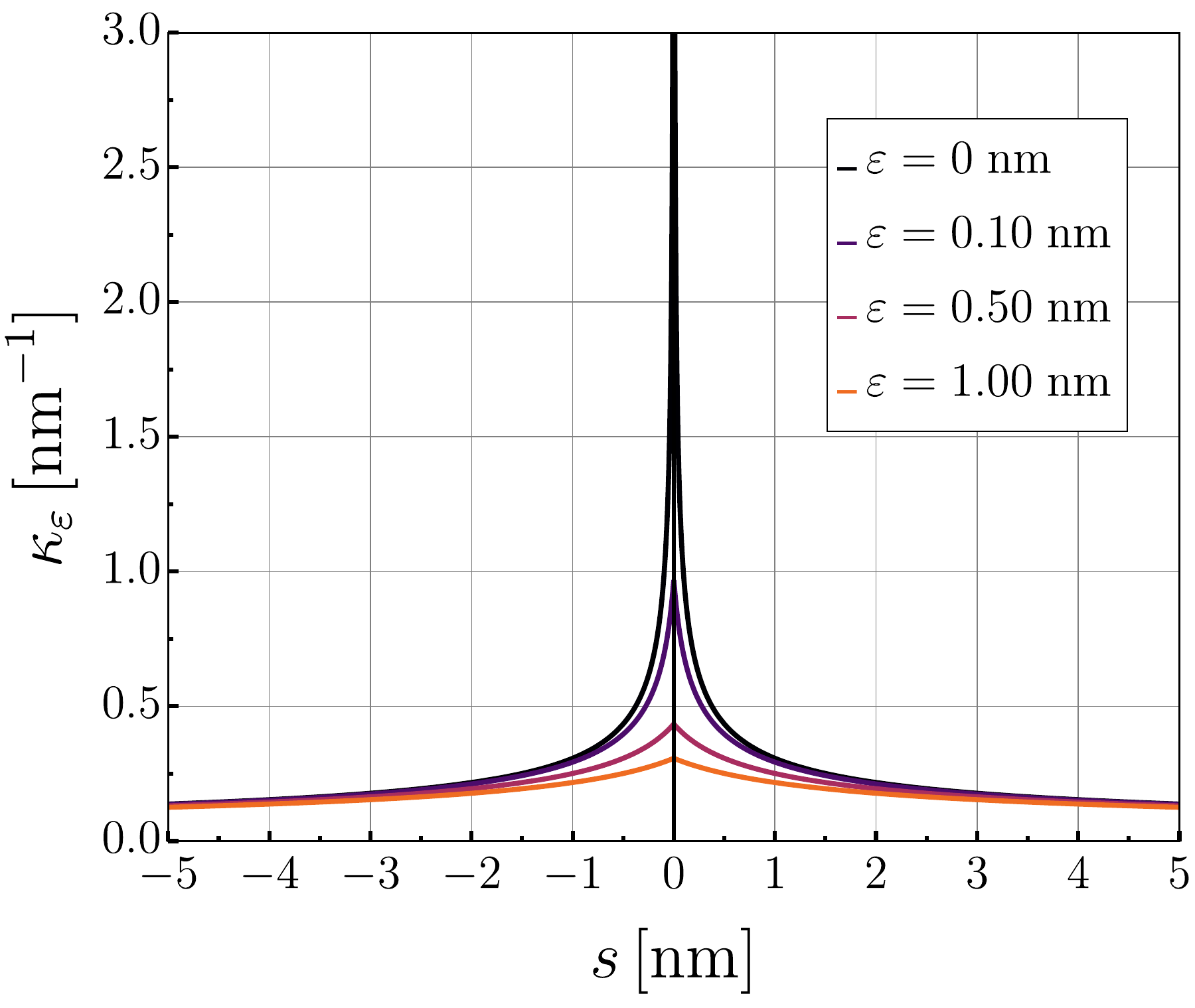}
    \caption{Curvature functions $\kappa_{\varepsilon}$ for four different representatives of the regularized curve family, and their divergent limit $\kappa$. The parameters are the same as in fig.\,\ref{fig:ExampleCurves}.}
    \label{fig:ExampleCurvature}
\end{figure}
\par
Using the regularization approach, we evaluated \eqref{eq:TangentialGleichung} numerically for the regular curves in dependence of $\varepsilon$, and calculated the eigenvalues with the corresponding eigenfunctions assuming Dirichlet boundary conditions. Therefore, the spectra are discrete. Figure \ref{fig:ExampleEnergies} shows the variance of the four lowest eigenvalues with the deformation parameter $\varepsilon$, and figure \ref{fig:ExampleWavefunctions} depicts the course of the probability density of the ground state wave function. The simulation confirms that in the limit $\varepsilon \rightarrow 0$ convergence of the eigenvalues and eigenfunctions occurs. In detail, we observe that the ground state energy lowers significantly and is the only one becoming negative. For our parameters, it reaches to an extrapolated value of approximately $-4.59\,\mathrm{meV}$. The corresponding wave function exhibits an increasingly narrow, and in the limit non-differentiable, peak centered on the degenerate locus, i.\,e. at $s=0$, where the geometric potentials are minimal. We emphasize that, since the eigenenergies converge to a finite value, this localization of the quantum state around the origin also stays finite. Thus, we obtained a solution for the quantum particle constrained to a degenerate curve.
\par
Interpreting the data as an approximation to the case of an infinitely large region, $J = \mathbb{R}$, we suspect that the ground state can be attributed the property of a bound state, which is geometrically induced by the sharp bent. The existence of such a state was proven \citep{Klaus1977, Goldstone1992, Exner1989} in the regular case. Then, the states of positive energy resemble a continuous part of the spectrum consisting of scattering states.
\begin{figure}
    \centering
    \includegraphics[width=\linewidth]{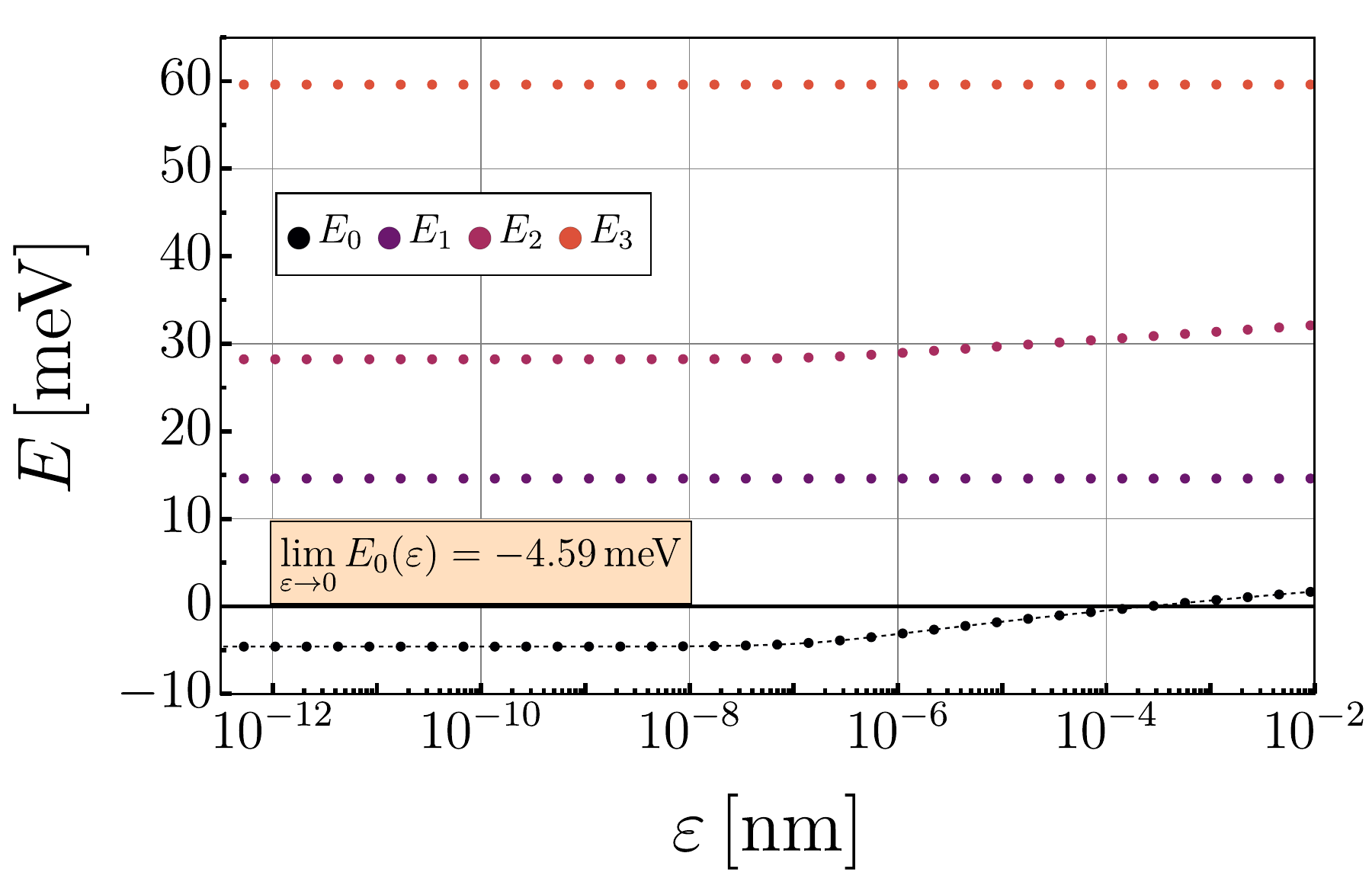}
    \caption{The four lowest energy eigenvalues of the regularized Hamilton operators in dependence of the de-singularization parameter $\varepsilon$. In the limit, the ground state is negative, and resembles a bound state when $J = \mathbb{R}$. The parameters are the same as in fig.\,\ref{fig:ExampleCurves}.}
    \label{fig:ExampleEnergies}
\end{figure}
\begin{figure}
    \centering
    \includegraphics[width=\linewidth]{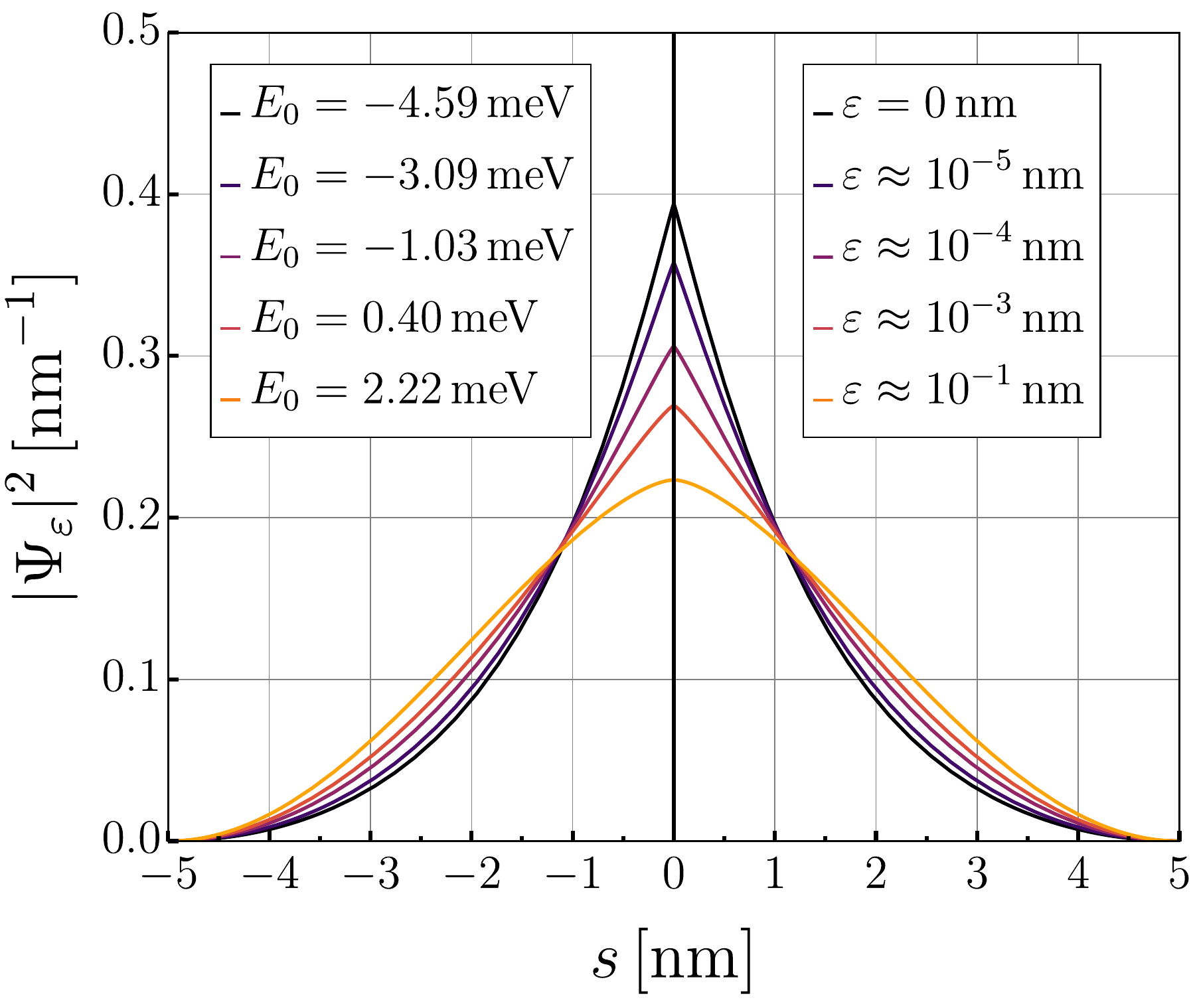}
    \caption{The probability density of the ground state for four regularized Hamilton operators, and their non-differentiable limit. The parameters are the same as in fig.\,\ref{fig:ExampleCurves}.}
    \label{fig:ExampleWavefunctions}
\end{figure}
\begin{figure}
    \centering
    \includegraphics[width=\linewidth]{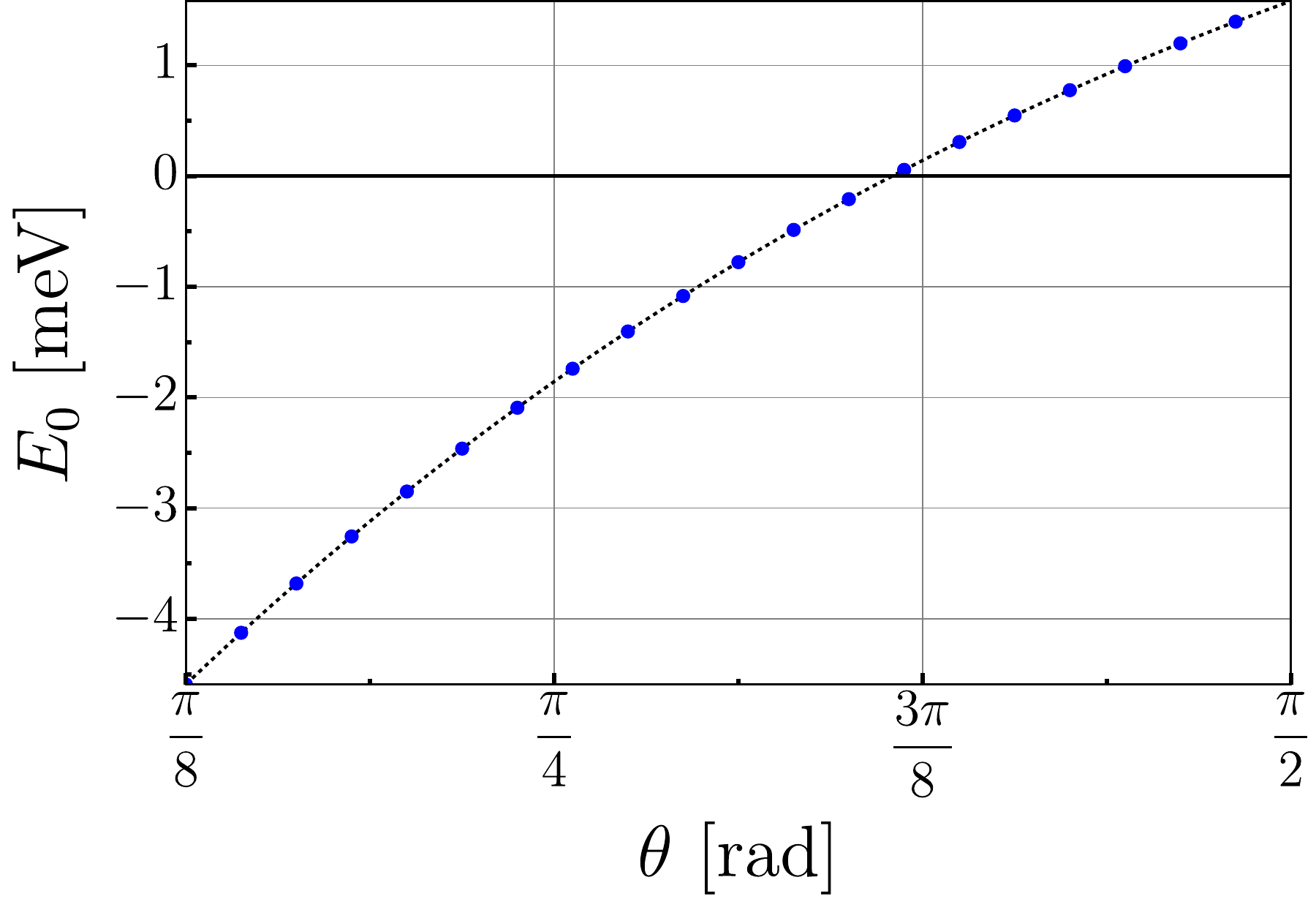}
    \caption{The ground state energy of the limiting degenerate case in dependence of the turn of the curve. Other parameters are the same as in fig.\,\ref{fig:ExampleCurves}.}
    \label{fig:ExampleAngleVariance}
\end{figure}
\FloatBarrier
\par
In addition, we investigated the sensitivity of the so obtained ground-state energy on the precise geometry of the degenerate curve by comparing otherwise similar structures with varying opening angle $\theta$. The result is depicted in figure \ref{fig:ExampleAngleVariance}. When the bending angle is increased, the ground state energy increases nonlinearly and ultimately becomes positive. The exact value should have an influence on e.\,g. transport processes or light-matter interaction, which conversely can be used for structural analysis of sharply bent quantum wires.
%
%%%%%%%%%%%%%%%%%%%%%%%%%%%%%%%%%%%%%%%%%%%%%%%%%%%%%%%%%%%%%%%%%%%%%%%%%%%%%%%%%%%%%%%%%%%%%%%%%%%%%%%%%%%%%%%%%%%%%%%%%%%%%%%%%%%
%%%%%%%%%%%%%%%%%%%%%%%%%%%%%%%%%%%%%%%%%%%%%%%%%%%%%%%%%%%%%%%%%%%%%%%%%%%%%%%%%%%%%%%%%%%%%%%%%%%%%%%%%%%%%%%%%%%%%%%%%%%%%%%%%%%
% Section V
\section{Conclusions}
\label{sec:Conclusion}
We discussed and presented  a CPA formalism to deal with a quantum mechanical particle under spatial irregular constraints and classified the various  cases, which we called degenerate submanifolds. For these cases,  we pointed out the need  for  an extensions of  established  CPA models. For the particular class of plane degenerate curves whose degenerate locus is defined by divergences of the curvature $\kappa\in L^{1}(J)$, we applied results from singular Sturm-Liouville theory and showed that if additionally $\kappa^{2}\in H^{-1}(J)$ holds, a family of regular curves can be found approximating the degenerate case both in the sense of subsets of Euclidean space and on the level of differential equations. Thus, the quantum state of the particle confined to the degenerate curve is well-defined as the limit of the converging eigensystems of the regularization.  The extended scheme is demonstrated  numerically showing that the singular curvatures in space can induce bound states with non-differentiable wave functions. 
So far the study is limited to an effective single particle model. It is established that in onedimensional systems electron-electron interaction leads to profound change in electronic and magnetic system properties \citep{giamarchi_book}. Constrained strongly correlated systems is a research topic yet to be addressed.
%
%%%%%%%%%%%%%%%%%%%%%%%%%%%%%%%%%%%%%%%%%%%%%%%%%%%%%%%%%%%%%%%%%%%%%%%%%%%%%%%%%%%%%%%%%%%%%%%%%%%%%%%%%%%%%%%%%%%%%%%%%%%%%%%%%%%
%%%%%%%%%%%%%%%%%%%%%%%%%%%%%%%%%%%%%%%%%%%%%%%%%%%%%%%%%%%%%%%%%%%%%%%%%%%%%%%%%%%%%%%%%%%%%%%%%%%%%%%%%%%%%%%%%%%%%%%%%%%%%%%%%%%
% Section A
\begin{appendices}
\begin{widetext}
\section{Mathematical Statements and Proofs}
Here, we give detailed mathematical comments on the statements made in the main text concerning the regularization ansatz, and some proofs.
\subsection{Convergence of Curves}
\label{sec:Appendix_Convergence_of_Curves}
\begin{theorem}
    Let $J = (a,b)$ be an open interval and $\mathcal{X}_{\mathcal{M}}$ a degenerate planar curve with a curvature $\kappa\in L^{1}(J)$. Further, let $\lbrace \kappa_\varepsilon \rbrace_{\varepsilon\in\R^{+}} \subset C^{0}(\overline{J})$ be a set of bounded, continuous functions that converge in $L^{1}(J)$ to $\kappa\in L^{1}(J)$, i.\,e. $\norm{\kappa_{\varepsilon} - \kappa}_{L^{1}(J)} \rightarrow 0$, $(\varepsilon\rightarrow 0)$.\\[5pt]
    Then the corresponding regular curves $\lbrace\mathcal{X}_{\mathcal{M}_{\varepsilon}}\rbrace_{\varepsilon\in \R^{+}}$, coinciding with $\mathcal{X}_{\mathcal{M}}$ at the initial value $s=a$ such that $\left[\mathcal{X}_{\mathcal{M}_{\varepsilon}}(a), \mathcal{X}'_{\mathcal{M}_{\varepsilon}}(a), \mathcal{X}''_{\mathcal{M}_{\varepsilon}}(a) \right] = \left[\mathcal{X}_{\mathcal{M}}(a), \mathcal{X}'_{\mathcal{M}}(a), \mathcal{X}''_{\mathcal{M}}(a) \right]$ for all $\varepsilon \in \mathbb{R}^{+}$, converge pointwise to the degenerate curve $\mathcal{X}_{\mathcal{M}}$, and the boundary values of the curves are understood in the sense of limits.
\end{theorem}
\begin{proof}
    For a $C^{2}$-curve $\function{\mathcal{X}_{\mathcal{M}_{\varepsilon}}}{J}{\R^{2}}$, $\varepsilon \in \mathbb{R}^{+}$, parametrized by arc length the Frenet-Serret equations read
    \begin{equation}
        \frac{\intd}{\intd s}\textbf{y}_{\varepsilon}(s) = \frac{\mathrm{d}}{\mathrm{d}s} \begin{pmatrix}
        \textbf{t}_{\varepsilon}(s) \\
        \textbf{n}_{\varepsilon}(s)
        \end{pmatrix}
        =
        \begin{pmatrix}
        0 & \kappa_{\varepsilon}(s) \\
        -\kappa_{\varepsilon}(s) & 0
        \end{pmatrix} \cdot \begin{pmatrix}
        \textbf{t}_{\varepsilon}(s) \\
        \textbf{n}_{\varepsilon}(s)
        \end{pmatrix}
        =
        \textbf{A}_{\varepsilon}(s)\cdot \textbf{y}_{\varepsilon}(s)\ ,\quad s\in J\ ,
    \end{equation}
    with the tangent and normal vector fields
    \begin{align}
        \begin{aligned}
            \textbf{t}_{\varepsilon}(s) = \mathcal{X}_{\mathcal{M}_{\varepsilon}}'(s)\ , \qquad 
            \textbf{n}_{\varepsilon}(s) = \frac{\mathcal{X}_{\mathcal{M}_{\varepsilon}}''(s)}{\norm{\mathcal{X}_{\mathcal{M}_{\varepsilon}}''(s)}_{2}} = \frac{1}{\kappa_{\varepsilon}(s)}\,\textbf{t}_{\varepsilon}'(s) \ .
        \end{aligned}
    \end{align}
    Now, let $\textbf{y} = (\textbf{t\ \ \textbf{n}})^{T}$ and $\textbf{A}$ denote the tuple of tangent and normal vector fields and the Frenet-Serret matrix corresponding to the similar problem describing the degenerate curve $\mathcal{X}_{\mathcal{M}}$, respectively. If $\kappa\in L^{1}(J)$, one can show by the theorem of Carath\'eodory that in this case there exists a unique solution in the extended sense $\textbf{y}$ with $\textbf{t},\textbf{n}\in \mathrm{AC}(J)$ in a neighborhood of the initial condition, so that $\mathcal{M}$ is a $C^{1}$-curve.
    \par
    In the first step we will show that the difference
    \begin{align}
        \begin{aligned}
            \delta \textbf{y}(\varepsilon; s) \coloneqq \textbf{y}_{\varepsilon}(s) - \textbf{y}(s)
            &=
            \delta \textbf{y}(\varepsilon;a) + \int\limits_{a}^{s} \left[ \textbf{A}_{\varepsilon}(t)\cdot \textbf{y}_{\varepsilon}(t)-\textbf{A}(t) \cdot \textbf{y}(t)\right] \intd t \\
            &=            
            \int\limits_{a}^{s} \left[ \textbf{A}_{\varepsilon}(t) \cdot \left( \textbf{y}_{\varepsilon}(t) - \textbf{y}(t)\right) + \left( \textbf{A}_{\varepsilon}(t) - \textbf{A}(t) \right) \cdot \textbf{y}(t) \right] \intd t \\
            &
            = \int\limits_{a}^{s} \textbf{A}_\varepsilon(t)\cdot \delta \textbf{y}(\varepsilon; t) \intd t + \int\limits_{a}^{s} \delta \textbf{A}(\varepsilon; t)\cdot \textbf{y}(t) \intd t \doteq I_{1}(s) + I_{2}(s)\ ,
        \end{aligned}
    \end{align}
    where $\delta \textbf{y}(\varepsilon;a) = \textbf{y}_{\varepsilon}(a)-\textbf{y}(a)=\textbf{0}$ by assumption, converges to zero int the pointwise sense as $\varepsilon \to 0$. Observing that
    \begin{align}
        \norm{\textbf{y}(s)}_{2} = \norm{\textbf{y}(a) + \int\limits_{a}^{s} \textbf{A}(t) \cdot \textbf{y}(t) \intd t}_{2} \leq \norm{\textbf{y}(a)}_{2} + \int\limits_{a}^{s} \norm{\textbf{A}(t) \cdot \textbf{y}(t)}_{2} \intd t = \norm{\textbf{y}(a)}_{2} + \int\limits_{a}^{s} \vert \kappa(t)\vert\,\norm{\textbf{y}(t)}_{2}\intd t\ ,
    \end{align}
    we reach to the following inequality by the lemma of Grönwall,
    \begin{align}
        \norm{\textbf{y}(s)}_{2} \leq \norm{\textbf{y}(a)}_{2}\,\left[ 1 + \int\limits_{a}^{s} \vert \kappa(t)\vert\,\mathrm{exp}\left( \int\limits_{t}^{s} \vert\kappa(\tau)\vert \intd \tau \right) \intd t \right]\ .
    \end{align}
    Varying over the entire interval $J$, this implies the existence of an upper bound on the solution $\textbf{y}$ which is given by
    \begin{align}
        \norm{\textbf{y}(s)}_{2} \leq \norm{\textbf{y}(a)}_{2}\,\left[ 1 + \int\limits_{a}^{b} \vert\kappa(t)\vert\,\mathrm{exp}\left( \int\limits_{a}^{b} \vert\kappa(\tau)\vert \intd \tau \right) \intd t \right] = \norm{\textbf{y}(a)}_{2}\,\left[ 1 + \Vert \kappa \Vert_{L^{1}(J)}\,\ee^{\Vert \kappa \Vert_{L^{1}(J)}} \right] \eqqcolon M\in \mathbb{R}^{+}\ .
    \end{align}
    With that, we get the convergence of $I_{2}$ as
    \begin{align}
        \Vert I_{2}(s)\Vert_{2} \leq \int\limits_{a}^{s} \Vert \delta\textbf{A}(\varepsilon; t)\cdot\textbf{y}(t)\Vert_{2} \intd t = \int\limits_{a}^{s} \vert \kappa_{\varepsilon}(t) - \kappa(t)\vert\, \Vert \textbf{y}(t)\Vert_{2} \intd t \leq M\,\Vert \kappa_{\varepsilon} - \kappa \Vert_{L^{1}(J)} \to 0\ ,\quad (\varepsilon \to 0)\ .
    \end{align}
    Now, we prove the convergence of the sum of both integrals, $I_{1} + I_{2}$. Because of
    \begin{align*}
        \Vert \textbf{A}_{\varepsilon}(t) \cdot \delta\textbf{y}(\varepsilon; t)\Vert_{2} &= \vert \kappa_{\varepsilon}\vert \,\Vert \delta\textbf{n}(\varepsilon;t) - \delta\textbf{t}(\varepsilon; t) \Vert_{2}\ ,\\
        \Vert \delta\textbf{n}(\varepsilon;t) - \delta\textbf{t}(\varepsilon; t) \Vert_{2}^{2} &= \Vert \delta\textbf{n}(\varepsilon;t)  \Vert_{2}^{2} - 2\,\langle \delta\textbf{t}(\varepsilon; t), \delta\textbf{n}(\varepsilon; t)\rangle_{2} + \Vert \delta\textbf{t}(\varepsilon; t) \Vert_{2}^{2} \leq \Vert \delta\textbf{t}(\varepsilon; t) \Vert_{2} + \Vert \delta\textbf{n}(\varepsilon; t) \Vert_{2}^{2} = \Vert \delta\textbf{y}(\varepsilon; t) \Vert_{2}^{2}
    \end{align*}
    such that
    \begin{align}
    \begin{aligned}
         \norm{\delta \textbf{y}(\varepsilon; s)}_{2} &\leq \norm{I_{1}(s)}_{2} + \norm{I_{2}(s)}_{2} \leq M\,\norm{\kappa_{\varepsilon} - \kappa}_{L^{1}(J)} + \int\limits_{a}^{s} \Vert \textbf{A}_{\varepsilon}(t) \cdot \delta\textbf{y}(\varepsilon; t)\Vert_{2} \intd t \\
         &\leq
         M\,\norm{\kappa_{\varepsilon} - \kappa}_{L^{1}(J)} + \int\limits_{a}^{s} \vert \kappa_{\varepsilon}(t)\vert\,\Vert \delta\textbf{y}(\varepsilon; t) \Vert_{2} \intd t\ ,
    \end{aligned}
    \end{align}
    the lemma of Grönwall again leads to the inequality
    \begin{align}
         \norm{\delta \textbf{y}(\varepsilon; s)}_{2} \leq M\,\norm{\kappa_{\varepsilon} - \kappa}_{L^{1}(J)}\,\left[ 1 + \int\limits_{a}^{s} \vert \kappa_{\varepsilon}(t)\vert\,\mathrm{exp}\left( \int\limits_{t}^{s} \vert \kappa_{\varepsilon}(\tau)\vert \intd \tau \right) \intd t \right]\ .
    \end{align}
    Again, by varying over the entire interval $J$ we obtain an upper bound on $\delta \textbf{y}[\varepsilon;\cdot]$, yielding the limit
    \begin{align}
    \begin{aligned}
        \norm{\delta \textbf{y}(\varepsilon; s)}_{2} &\leq M\,\norm{\kappa_{\varepsilon} - \kappa}_{L^{1}(J)}\,\left[ 1 + \int\limits_{a}^{b} \vert \kappa_{\varepsilon}(t)\vert\,\mathrm{exp}\left( \int\limits_{a}^{b} \vert\kappa_{\varepsilon}(\tau)\vert \intd \tau \right) \intd t \right] \\
        &= M\,\norm{\kappa_{\varepsilon} - \kappa}_{L^{1}(J)}\,\left[ 1 + \Vert \kappa_{\varepsilon} \Vert_{L^{1}(J)}\,\ee^{\Vert \kappa_{\varepsilon} \Vert_{L^{1}(J)}} \right]\\
        &\eqqcolon
        M\,M_{\varepsilon}\,\norm{\kappa_{\varepsilon} - \kappa}_{L^{1}(J)} \to 0\ ,\quad (\varepsilon \to 0)\ .
    \end{aligned}
    \end{align}
    We set the difference in of the curve points associated with a specific value of $s$ as $\delta\mathcal{X}_{\mathcal{M}}(\varepsilon;\cdot) \coloneqq \mathcal{X}_{\mathcal{M}_{\varepsilon}} - \mathcal{X_{\mathcal{M}}}$. Because of $\delta\mathcal{X}_{\mathcal{M}}'(\varepsilon; \cdot) = \delta \textbf{t}(\varepsilon; \cdot)$ and $\delta\mathcal{X}_{\mathcal{M}}(\varepsilon;a) = 0$ by assumption, we finally obtain for every $s\in J$
    \begin{align}
    \begin{aligned}
        \norm{\delta\mathcal{X}_{\mathcal{M}}(\varepsilon; s)}_{2} &= \norm{\delta\mathcal{X}_{\mathcal{M}}(\varepsilon; a) + \int\limits_{a}^{s} \delta\mathcal{X}_{\mathcal{M}}'(\varepsilon; t) \intd t}_{2} \leq \int\limits_{a}^{s} \norm{\delta\textbf{y}(\varepsilon; t)}_{2}\intd t \\
        &\leq M\,M_{\varepsilon}\,(b-a)\,\norm{\kappa_{\varepsilon} - \kappa}_{L^{1}(J)} \to 0\ ,\quad (\varepsilon \to 0)\ .
    \end{aligned}
    \end{align}
    Thus, the family of curves converges in the pointwise sense.
\end{proof}
\subsection{Convergence of Spectra and Eigenfunctions}
\label{sec:Appendix_Convergence_of_Spectra}
In what follows, we assume that $I = (a,b) \subset \mathbb{R}$ is an open and bounded interval. Furthermore, $W^{1,1}(I)$ is the Sobolev space of Lebesgue integrable functions on $I$ possessing one weak derivative (that is also Lebesgue integrable), which in one dimension coincides with the space of absolutely continuous functions, i.\,e. $W^{1,1}(I) = \mathrm{AC}(I)$. As we described in the main text, extending the CPA formalism to a degenerate planar curve containing a degenerate locus $\Xi_{2,1}(\mathcal{M})$ that is given by divergences in the curvature results in the encounter of a singular potential-like correction term (the geometric potential). In the case of curves, we thus find a singular Sturm-Liouville problem. In the context of quantum mechanics, the Hamilton operator is demanded to be a self-adjoint endomorphism on a Hilbert space, and for this purpose we canonically choose $\left( L^{2}(I), \langle \cdot, \cdot\rangle_{L^{2}(I)} \right)$.
\par
We remark that there are at least two ways of solving this singular Sturm-Liouville equation we obtained. In one approach, removing the degenerate locus and considering the regular part only dissects $\mathcal{M}$ into several connected components and creates a multi-interval Sturm-Liouville problem. If we assume that there is only one isolated singularity, this would correspond to constructing a singular two-interval Sturm-Liouville boundary value problem, which is a subject where many recent results were obtained (see for example \citep{Wang2007, Sun2007, Cao2009, Zettl2021}). Constructing the model this way appears as a canonical approach. 
\par
Nevertheless, there is a different approach that can be applied to certain subclasses of singular problems, including the special case of dealing with a distribution of first order as a potential. It relies on redefining the quasi-derivatives and Lagrange boundary forms considered in Sturm-Liouville theory, rather than dissecting the configuration manifold. These adaptions allow to treat the problem as a regular (in the sense of Birkhoff) one-interval Sturm-Liouville boundary value problem, and was developed in \citep{A.M.Savchuk1999, Savchuk2003}. See also \citep{Eckhardt2013} and the references therein. In particular, it was shown that the solution of the singular problem can be calculated as the limit of the solutions of a one-parameter family of regularized problems. Here, we also focus on this approach and summarize some important points below.
\par
At first, we recall the regular case. Let there be a family of real functions $\left\{ u_{\varepsilon} \right\}_{\varepsilon \in \R^{+}} \subset C^{1}(I)$. Taking the classical derivative to set $v_{\varepsilon} = u_{\varepsilon}'$, we obtain a family of operators $\{ (\hat{L}_{\varepsilon}, \mathcal{D}(\hat{L}_{\varepsilon}) \}_{\varepsilon \in \R^{+}}$ defined by
\begin{align}
    \hat{L}_{\varepsilon} \coloneqq -\frac{\mathrm{d}^{2}}{\mathrm{d}x^{2}} + v_{\varepsilon}(x)\,\hat{I}\ .
\end{align}
As the potential $v_{\varepsilon} \in C^{0}(I) \subset L^{\infty}(I)$ is continuous and hence especially bounded, the associated maximal and minimal operators, $(\hat{L}_{\varepsilon}, \mathcal{D}_{\varepsilon}^{\mathrm{max}})$ and $(\hat{L}_{\varepsilon}, \mathcal{D}_{\varepsilon}^{\mathrm{min}})$ respectively, are constructed in the classical way \citep{Zettl2005}. The maximal domain associated with the mapping $\hat{L}_{\varepsilon}^{\mathrm{max}}: \mathcal{D}_{\varepsilon}^{\mathrm{max}} \to L^{2}(I)$ is given as
\begin{align}
    \mathcal{D}_{\varepsilon}^{\mathrm{max}} \overset{:}{=} \mathcal{D}(\hat{L}_{\varepsilon}^{\mathrm{max}}) = \left\{ \varphi \in L^{2}(I): \varphi, \varphi'\in W^{1,1}(I)\ \wedge\ \hat{L}_{\varepsilon}\varphi \in L^{2}(I) \right\} = H^{2}(I) \quad \forall\ \varepsilon\in\R^{+}\ ,
\end{align}
whereas for the minimal domain associated with the mapping $\hat{L}_{\varepsilon}^{\mathrm{min}}: \mathcal{D}_{\varepsilon}^{\mathrm{min}} \to L^{2}(I)$ we have the closure of $C_{c}^{\infty}(I)$ in $H^{2}(I)$. So we obtain the restriction of the maximal operator to a domain of functions vanishing at the boundary in the sense of traces,
\begin{align}
\begin{aligned}
     \mathcal{D}_{\varepsilon}^{\mathrm{min}} \overset{:}{=} \mathcal{D}(\hat{L}_{\varepsilon}^{\mathrm{min}}) &= \left\{ \varphi \in \mathcal{D}_{\varepsilon}^{\mathrm{max}}:\ \exists\,(h_{n})_{n\in\mathbb{N}}\subset C_{c}^{\infty}(I): \lim_{n\to\infty} \Vert h_{n} - \varphi \Vert_{H^{2}(I)} = 0 \right\} \\
     &= \left\{ \varphi \in \mathcal{D}_{\varepsilon}^{\mathrm{max}}: \varphi\vert_{\partial I} = 0\ \wedge\ \varphi'\vert_{\partial I} = 0 \right\} = H^{2}_{0}(I) \quad \forall\ \varepsilon \in \R^{+}\ .
\end{aligned}
\end{align}
For every $\varepsilon \in \R^{+}$ it holds that $( \hat{L}_{\varepsilon}^{\mathrm{max}} )^{\dagger} = \hat{L}_{\varepsilon}^{\mathrm{min}}$ and $( \hat{L}_{\varepsilon}^{\mathrm{min}} )^{\dagger} = \hat{L}_{\varepsilon}^{\mathrm{max}}$, which can be seen from the Lagrange boundary form
\begin{align}
    \mathcal{L}(\varphi, \psi) \coloneqq \langle \hat{L}_{\varepsilon}\varphi, \psi \rangle_{L^{2}(I)} - \langle \varphi, \hat{L}_{\varepsilon}\psi \rangle_{L^{2}(I)} = \left[ {\psi^{\ast}}'\,\varphi - \psi^{\ast}\,\varphi'  \right]_{a}^{b}\ .
    \label{eq:Lagrange_form}
\end{align}
In the context of quantum mechanics, the desired operator is a self-adjoint realization fulfilling
\begin{align}
    \hat{L}_{\varepsilon}^{\mathrm{min}} \subset \hat{L}_{\varepsilon} = \hat{L}_{\varepsilon}^{\dagger} \subset \hat{L}_{\varepsilon}^{\mathrm{max}}
\end{align}
in the sense of operator extensions.
\par
It is a well-known fact that in Sturm-Liouville theory the self-adjoint realizations correspond uniquely to different boundary conditions, see \citep{Zettl2005, Zettl2021} and references therein. Precisely, if we define
\begin{align}
    \mathbf{\Phi}(x) \coloneqq \begin{pmatrix}
        \varphi(x) \\
        \varphi'(x)
    \end{pmatrix}
\end{align}
and consider the regular two-point boundary value problem
\begin{align}
    \begin{cases}
        \hat{L}_{\varepsilon}\varphi(x) =  -\varphi''(x) + v_{\varepsilon}(x)\,\varphi(x) = \lambda\,\varphi(x) &: x\in I \\
        \hspace{1.755 cm} A\cdot \mathbf{\Phi}(a) + B\cdot \mathbf{\Phi}(b) = \mathbf{0}
    \end{cases}\ ,
\end{align}
the boundary conditions can equivalently be expressed via the boundary linear forms
\begin{align}
    \hat{U}_{j}[\varphi] \coloneqq A_{j1} \varphi(a) + A_{j2} \varphi'(a) + B_{j1}\varphi(b) + B_{j2}\varphi'(b) = 0\ ,\quad j\in \{ 1;2 \}\ .
\label{eq:boundary_linear_form}
\end{align}
In order to deliver a self-adjoint realization, the coefficient matrices $A,B\in \mathbb{C}^{2\times2}$ are demanded to fulfill the conditions
\begin{align}
\begin{aligned}
    \mathrm{(i)}&\ A \cdot E\cdot A^{\dagger} = B \cdot E\cdot B^{\dagger} \quad\text{with}\quad E = \begin{pmatrix}
        0 & -1 \\
        1 & 0
    \end{pmatrix} \quad \Leftrightarrow\quad A_{j1}^{\ast}A_{k2} - A_{k1}A_{j2}^{\ast} = B_{j1}^{\ast}B_{k2} - B_{k1}B_{j2}^{\ast}\ ,\quad j,k\in \{ 1;2 \}\ , \\[5pt]
    \mathrm{(ii)}&\ \mathrm{rank}(A \vert B) = 2\ .
\end{aligned}
\label{eq:conditions_sa}
\end{align}
Then, any self-adjoint realization $(\hat{L}_{\varepsilon}, \mathcal{D}(\hat{L}_{\varepsilon}))$ is readily obtained by specifying the domain as
\begin{align}
    \mathcal{D}(\hat{L}_{\varepsilon}) \coloneqq \left\{ \varphi \in \mathcal{D}^{\mathrm{max}}_{\varepsilon}: U_{j}[\varphi] = 0 \ \ \forall\,j\in \{ 1;2 \} \right\} ,
\end{align}
with $U_{j}$, $j\in \{1;2\}$ respecting \eqref{eq:conditions_sa}, as then the Lagrange boundary form \eqref{eq:Lagrange_form} vanishes while $\varphi$ and $\psi$ are taken from the same domain.
\par
Now, we consider the singular case where the operator definition
\begin{align}
    \hat{L} \coloneqq -\frac{\mathrm{d}^{2}}{\mathrm{d}x^{2}} + v(x)\,\hat{I}
\end{align}
contains a singular real potential which is assumed to fulfill $v = u'$ for some function $u\in L^{2}(I)$ in the distributional sense, i.\,e. $\langle u , h' \rangle_{L^{2}(I)} = - \langle v, h\rangle_{L^{2}(I)}\ \forall\ h\in \mathcal{D}(I) \doteq C_{c}^{\infty}(I)$. This corresponds to $v \in H^{-1}(I)$, see \citep{Adams2003, Hryniv2006, Albeverio2008}. In \citep{A.M.Savchuk1999} an adapted quasi-derivative of a function was introduced,
\begin{align}
     \varphi^{[1]} \coloneqq \varphi' - u\,\varphi\ ,
     \label{eq:quasiderivative}
\end{align}
so that the operator action can be redefined as
\begin{align}
    \hat{L}\varphi = -(\varphi^{[1]})' - u\,\varphi^{[1]} - u^{2}\,\varphi
\end{align}
and the Lagrange boundary form adapts to
\begin{align}
    \tilde{\mathcal{L}}(\varphi, \psi) \coloneqq \langle \hat{L}\varphi, \psi\rangle_{L^{2}(I)} - \langle \varphi, \hat{L}\psi \rangle_{L^{2}(I)} = \left[ (\psi^{\ast})^{[1]}\,\varphi - \psi^{\ast}\varphi^{[1]} \right]_{a}^{b}\ .
\end{align}
In these expressions, no distribution-valued terms appear anymore. Then, many relations of the regular case can be generalized and the singular case can be shown to possess similar properties when the weak derivatives and the usual quasi-derivatives used in Sturm-Liouville theory get replaced by this quasi-derivative \eqref{eq:quasiderivative} \citep{A.M.Savchuk1999, Savchuk2003, Eckhardt2013}. Precisely, the maximal operator $(\hat{L}, \mathcal{D}^{\mathrm{max}})$ is associated with the mapping $\hat{L}^{\mathrm{max}}: \mathcal{D}^\mathrm{max} \to L^{2}(I)$ with the maximal domain
\begin{align}
    \mathcal{D}^{\mathrm{max}} \overset{:}{=} \mathcal{D}(\hat{L}^{\mathrm{max}}) = \left\{ \varphi\in L^{2}(I): \varphi,\varphi^{[1]}\in W^{1,1}(I) \ \wedge\ \hat{L}\varphi\in L^{2}(I)  \right\}\ ,
\end{align}
and the minimal operator $(\hat{L}^{\mathrm{min}}, \mathcal{D}(\hat{L}^{\mathrm{min}}))$ with the mapping $\hat{L}^{\mathrm{min}}: \mathcal{D}^\mathrm{min} \to L^{2}(I)$ incorporating the minimal domain
\begin{align}
     \mathcal{D}^{\mathrm{min}} \overset{:}{=} \mathcal{D}(\hat{L}^{\mathrm{min}}) = \left\{ \varphi\in \mathcal{D}^{\mathrm{max}}: \varphi\vert_{\partial I} = 0\ \wedge\ \varphi^{[1]}\vert_{\partial I} = 0 \right\}\ .
     \label{eq:singular_minimal_domain}
\end{align}
Any self-adjoint realization $(\hat{L}, \mathcal{D}(\hat{L}))$ is determined by a domain
\begin{align}
    \mathcal{D}(\hat{L}) \coloneqq \left\{ \varphi \in \mathcal{D}^{\mathrm{max}}: \tilde{U}_{j}[\varphi] = 0 \ \ \forall\,j\in \{ 1;2 \} \right\} ,
\end{align}
where the boundary linear forms $\tilde{U}_{j}$, $j\in \{1;2\}$ express the general form of boundary conditions similar to \eqref{eq:boundary_linear_form} - \eqref{eq:conditions_sa}, just with adapted coefficient matrices. In the regular case the operator definitions based on both derivatives (after taking the closure of the minimal operator) are equivalent in the following sense: If it holds that $u' = v \in L^{1}(I)$, then $\varphi'\in W^{1,1}(I)\ \Leftrightarrow\ \varphi^{[1]}\in W^{1,1}(I)$. Thus, the Lagrange boundary forms and all statements leading to the general classifications of self-adjoint boundary conditions are equivalent. However, the operator definitions based on the quasi-derivative \eqref{eq:quasiderivative} still hold for the singular case. This is the basis for the regularization approach whose idea is to characterize the spectrum of the singular operator $(\hat{L}, \mathcal{D}(\hat{L}))$ by replacing it with a family of operators $\{ (\hat{L}_{\varepsilon}, \mathcal{D}(\hat{L}_{\varepsilon}) \}_{\varepsilon \in \R^{+}}$ according to the above description and taking the limit $\varepsilon \to 0$. All operators are understood as self-adjoint realizations.
\par
Nevertheless, note that the domains of the regularized operators and their singular counterparts generally do not coincide. Rather, the domains of the singular operators possess less structure as one can only show \citep{Hryniv2006, Albeverio2008} that $\mathcal{D}^{\mathrm{max}} \subset H^{1}(I)$ whereas $\mathcal{D}_{\varepsilon}^{\mathrm{max}} = H^{2}(I)$. Although there is a continuous embedding $H^{2}(I) \hookrightarrow H^{1}(I)$ of a dense subspace by the Sobolev embedding theorem, it is generally not true that one domain is included in the other, so we have $\mathcal{D}^{\mathrm{max}} \not\subset \mathcal{D}_{\varepsilon}^{\mathrm{max}}$ and $\mathcal{D}^{\mathrm{max}} \not\supset \mathcal{D}_{\varepsilon}^{\mathrm{max}}$. This means that in the limiting singular case neither do the wave functions have to be two times differentiable functions nor can the Hamilton operator be separated in kinetic and potential energies as usual because neither of these needs to be an individual observable (i.\,e. a self-adjoint endomorphism on $\left( L^{2}(I), \langle \cdot, \cdot \rangle_{L^{2}(I)} \right)$ with proper $L^{2}(I)$-domain).
\par
It is well-known that $\mathcal{D}_{\varepsilon}^{\mathrm{min}} = H^{2}_{0}(I)\subset L^{2}(I)$ is a dense subspace. In \citep{A.M.Savchuk1999} it is stated that the same is true for $\mathcal{D}^{\mathrm{min}} \subset L^{2}(I)$. Furthermore \citep{A.M.Savchuk1999}, it is known that in the case of real potentials all $(\hat{L}_{\varepsilon}, \mathcal{D}(\hat{L}_{\varepsilon}))$, $\varepsilon \in \mathbb{R}^{+}$, and $(\hat{L}, \mathcal{D}(\hat{L}))$ with domains defined as above are (closed) self-adjoint linear operators on the Hilbert space $\left( L^{2}(I), \langle \cdot, \cdot \rangle_{L^{2}(I)} \right)$. So by the spectral theorem all spectra are real, $\sigma(\hat{L}_{\varepsilon}),\sigma(\hat{L}) \subset \mathbb{R}$. In the regular case \citep{Zettl2005, Zettl2021}, all the spectra $\sigma(\hat{L}_{\varepsilon})$, $\varepsilon\in \mathbb{R}^{+}$, are discrete because of the boundary conditions incorporated in the respective domains $\mathcal{D}(\hat{L}_{\varepsilon})$. Precisely, there is an infinite but countable number of eigenvalues without any accumulation point. In the singular case, it was shown \citep{A.M.Savchuk1999, Savchuk2003} that the self-adjoint boundary conditions imply that the same is true for the spectrum $\sigma(\hat{L})$, because the resolvent operator is compact outside of the spectrum. The discreteness of the spectra is equivalent to the fact that the respective sets of eigenfunctions given as solutions of the eigenvalue equations
\begin{align}
\begin{aligned}
    \hat{L}_{\varepsilon}\varphi_{\varepsilon}^{n} &= \lambda_{\varepsilon}^{n}\varphi_{\varepsilon}^{n}\ ,\quad \varepsilon\in \mathbb{R}^{+}\ , \\
    \hat{L}\varphi^{n} &= \lambda^{n} \varphi^{n}\ ,
\end{aligned}
\end{align}
form orthonormal bases of $\left( L^{2}(I), \langle \cdot, \cdot \rangle_{L^{2}(I)} \right)$ with $n\in\mathbb{N}$ counting the eigenstates without multiplicity. Especially, the sets of eigenfunctions form Riesz bases. For any regular self-adjoint Sturm-Liouville operator, the multiplicity of every eigenvalue is less than or equal to two.
\par
Furthermore, we also state that all the spectra $\{ \sigma(\hat{L}_{\varepsilon})\}_{\varepsilon \in\mathbb{R}^{+}}$ as well as $\sigma(\hat{L})$ are semi-bounded from below. While this fact is well-known for the regularized cases \cite{Zettl2005, Zettl2021} it needs to be verified for the singular case. This was achieved in \citep{Eckhardt2013} for an even more general setting.
\par
Altogether, the following theorem (see below) holds, where we want to add a corollary to the previous considerations \citep{A.M.Savchuk1999, Savchuk2003} which explicitly states that invariant subspaces generated by the eigenfunctions of the regularized operators are included in those generated by the eigenfunctions of the singular operator, where a correspondence is set by the convergence of the eigenvalues.
\begin{theorem}
\label{Th:ConvergenceEigenfunctions}
Let $I = (a,b)\subset \mathbb{R}$ and $(\hat{L}, \mathcal{D}(\hat{L}))$ the operator with the associated mapping $\hat{L}: \mathcal{D}(\hat{L}) \to L^2(I)$ and the domain
\begin{align*}
\begin{aligned}
    \mathcal{D}(\hat{L}) &= \lbrace \varphi \in L^{2}(I) :\  \varphi,\varphi^{[1]} \in W^{1,1}(I) \ \wedge\ \tilde{U}_{1}[\varphi] = \tilde{U}_{2}[\varphi] = 0 \rbrace\ ,
\end{aligned}
\end{align*}
with $\tilde{U}_{1}$ and $\tilde{U}_{2}$ being linear forms encoding self-adjoint boundary conditions, and the definition
\begin{align*}
    \hat{L} \coloneqq -\frac{\mathrm{d}^2}{\mathrm{d}x^{2}} + v(x)\,\hat{I} \ ,
\end{align*}
where $v = u'$ for some primitive $u\in L^{2}(I)$ is a real-valued distribution of first order on $I$. Furthermore, let the family $\left\{ u_{\varepsilon}\right\}_{\varepsilon \in \mathbb{R}^{+}} \subset C^{1}(I)$ be a regularization on $I$ such that
\begin{align*}
    \norm{u_{\varepsilon} - u}_{L^{2}(I)} \to 0\ ,\quad (\varepsilon \to 0)\ .
\end{align*}
Then, setting the real functions $v_{\varepsilon} \coloneqq u'_{\varepsilon}$, there is a family of operators $\{ (\hat{L}_{\varepsilon} ,\mathcal{D}(\hat{L}_{\varepsilon})\}_{\varepsilon\in\mathbb{R}^{+}}$ with mappings \linebreak $\hat{L}_{\varepsilon}:\mathcal{D}(\hat{L}_{\varepsilon})\rightarrow L^2(I)$ all having the domain
\begin{align*}
    \mathcal{D}(\hat{L}_{\varepsilon}) = \left\{ \varphi \in H^{2}(I): \tilde{U}_{1}[\varphi] = \tilde{U}_{2}[\varphi] = 0 \right\} \quad \forall\ \varepsilon\in\mathbb{R}^{+}
\end{align*}
and obeying the definitions
\begin{align*}
    \hat{L}_{\varepsilon} \coloneqq -\frac{\mathrm{d}^{2}}{\mathrm{d}x^{2}}  + v_{\varepsilon}(x)\,\hat{I}\ .
\end{align*}
Assuming that each eigenvalue $\lambda^{n} \in \sigma(\hat{L})$ of $\hat{L}$ (counted by $n\in\mathbb{N}$ without multiplicities) has a corresponding multiplicity $m(\lambda^{n})\in \mathbb{N}$ such that no eigenvalue deficiencies occur, the following statements hold:
\begin{enumerate}
    \item $\{ (\hat{L}_{\varepsilon}, \mathcal{D}(\hat{L}_{\varepsilon})) \}_{\varepsilon\in\mathbb{R}^{+}}$ converges in the norm resolvent sense to $(\hat{L}, \mathcal{D}(\hat{L}))$ as $\varepsilon \to 0$.
    \item The family of spectra $\{ \sigma(\hat{L}_{\varepsilon}) \}_{\varepsilon \in \R^{+}}$ converges to $\sigma(\hat{L})$. The spectrum $\sigma(\hat{L})$ is real, discrete, and semi-bounded from below. The eigenvalues $\lambda_{\varepsilon}^{l_{n}}$ of the mappings $\hat{L}_{\varepsilon}$ converge to a corresponding eigenvalue $\lambda^{n}$ of $\hat{L}$ according to 
    \begin{align*}
        \left\vert \lambda_{\varepsilon}^{l_{n}} - \lambda^{n} \right\vert \leq C(n)\,\norm{ u_{\varepsilon} - u}_{L^{2}(I)}\ ,
    \end{align*}
    where $l_{n} \in\mathcal{I}_{n} \subset \mathbb{N}$ incorporates the partial splitting of degeneracy of the state defined by $\lambda^{n}$ during the regularization.
    \item The corresponding eigenspaces (invariant subspaces) associated with the $\lambda_{\varepsilon}^{l_{n}}$ are included in that of $\lambda^{n}$ as $\varepsilon \to 0$.
\end{enumerate}
\end{theorem}
\begin{proof}
    The statements one and two have been proven in \citep{A.M.Savchuk1999, Savchuk2003, Eckhardt2013} and related works. Here, we only want to focus on the third statement which follows as a corollary of the first two by examining the convergence of Riesz projections.\\
    In \citep{A.M.Savchuk1999} it is stated that $\mathcal{D}(\hat{L}) \subset L^{2}(I)$ is a dense subspace. As it is well-known that $H^{2}(I)\subset L^{2}(I)$ is a dense subspace, the same is true for $\mathcal{D}(\hat{L}_{\varepsilon}) \subset L^{2}(I)$. The eigenfunctions of the regular and singular operators, defined by
    \begin{align}
    \begin{aligned}
        \hat{L}_{\varepsilon}\varphi_{\varepsilon}^{n} &= \lambda_{\varepsilon}^{n}\varphi_{\varepsilon}^{n}\ ,\quad \varepsilon\in \mathbb{R}^{+}\ , \\
        \hat{L}\varphi^{n} &= \lambda^{n} \varphi^{n}
    \end{aligned}
    \end{align}
    and counted by $n\in\mathbb{N}$ without multiplicity, form orthonormal bases of $\left( L^{2}(I), \langle \cdot, \cdot \rangle_{L^{2}(I)} \right)$ and, in particular, Riesz bases.\\     
    Consider the $n$-th eigenvalue $\lambda^{n} \in \sigma(\hat{L})$ (counted without multiplicity). As under the given prerequisites the spectra of the regularized operators converge according to statements one and two, for every $\varepsilon \in \mathbb{R}^{+}$ there is a set of eigenvalues $\{\lambda_{\varepsilon}^{l_{n}}\}_{l_{n}\in \mathcal{I}_{n}} \eqqcolon \Omega_{\varepsilon}(n) \subset \sigma(\hat{L}_{\varepsilon})$ fulfilling $\vert \lambda_{\varepsilon}^{l_{n}} - \lambda^{n} \vert_{\R} \to 0$, $(\varepsilon \to 0)$. Thus, we can choose some $\varepsilon_{0}(n)\in\R^{+}$ such that
    \begin{align}
        \vert \lambda_{\varepsilon}^{l_{n}} - \lambda^{n} \vert_{\R} < \frac{1}{2} \min\limits_{m\in \mathbb{N}: m\neq n}\vert \lambda^{m} - \lambda^{n} \vert_{\R} \eqqcolon R(n) \quad \forall\,l_{n}\in \mathcal{I}_{n}\ \forall\, \varepsilon\in\R^{+}: \varepsilon \leq \varepsilon_{0}(n)\ .
    \end{align}
    As all spectra are real, a suitable curve to define the Riesz projectors is given by
    \begin{align}
        \label{eq:ProjectorCurve}
        \Gamma(n) = \left\lbrace z\in\C: z = \lambda^{n} + R(n)\,\ee^{\ii \alpha}\ \vert\ \alpha \in [0, 2\pi) \right\rbrace\ .
    \end{align}
    We observe that the Riesz projectors $\hat{P}_{\lambda^{n}}$ and $\hat{P}_{\Omega_{\varepsilon}(n)}$ are well-defined due to the discreteness of the spectra. Thus, for sufficiently small $\varepsilon\in \R^{+}: \varepsilon \leq \varepsilon_{0}(n)$ we get
    \begin{align}
    \label{eq:ProjectorConvergence}
    \begin{aligned}
        \norm{\hat{P}_{\Omega_{\varepsilon}(n)} - \hat{P}_{\lambda^{n}}} &= \norm{-\frac{1}{2\pi\ii} \oint\limits_{\Gamma(n)} \left( \left[ \hat{L}_{\varepsilon} - z\,\hat{I} \right]^{-1} - \left[\hat{L} - z\,\hat{I} \right]^{-1} \right)\intd z} 
        \leq
        \frac{1}{2\pi} \oint\limits_{\Gamma(n)} \norm{\left[\hat{L}_{\varepsilon} - z\,\hat{I}\right]^{-1} - \left[\hat{L} - z\,\hat{I} \right]^{-1}}\intd z \\
        &\leq
        \frac{\vert \Gamma(n)\vert}{2\pi}\,\sup_{z\in\Gamma(n)} \norm{ \left[\hat{L}_{\varepsilon} - z\,\hat{I}\right]^{-1} - \left[ \hat{L} - z\,\hat{I} \right]^{-1}} 
        \rightarrow
        0\ ,\quad (\varepsilon\rightarrow 0)\ ,
    \end{aligned}
    \end{align}
    because the family of compact resolvent operators $\hat{R}_{\varepsilon}(z) \coloneqq [\hat{L}_{\varepsilon} - z\,\hat{I}]^{-1}$ converges in the norm topology to the resolvent $\hat{R}(z) \coloneqq [\hat{L} - z\,\hat{I}]^{-1}$, which is hence also compact (this is the proof of \citep{A.M.Savchuk1999}, Theorem 3; see also the proof of \citep{A.M.Savchuk1999}, Lemma 5). With that, the definition of the operator norm implies for an arbitrary wave function
    \begin{align}
        \forall\ \Psi\in \mathcal{D}(\hat{L}): \left\Vert \left(\hat{P}_{\Omega_{\varepsilon}(n)} - \hat{P}_{\lambda^{n}} \right) \Psi\right\Vert_{L^{2}(I)} \leq \left\Vert \hat{P}_{\Omega_{\varepsilon}(n)} - \hat{P}_{\lambda^{n}} \right\Vert\cdot\Vert\Psi\Vert_{L^{2}(I)} \to 0\ ,\quad (\varepsilon \to 0)\ .
    \end{align}
    From this, it follows
    \begin{align}
        \lim_{\varepsilon\to 0} \hat{P}_{\Omega_{\varepsilon}(n)} = \hat{P}_{\lambda^{n}} \quad \Rightarrow\quad \forall\ \Psi\in \mathcal{D}(\hat{L}):\ \lim_{\varepsilon\to 0} \hat{P}_{\Omega_{\varepsilon}(n)}\Psi = \hat{P}_{\lambda^{n}}\Psi\ ,
    \end{align}
    where the each limit is taken in the respective sense. As $\Psi\in \mathcal{D}(\hat{L})$ is arbitrary, this means that every element of the respective eigenspace associated with $\Omega_{\varepsilon}(n) \subset \sigma(\hat{L}_{\varepsilon})$ converges to a function that lies inside the eigenspace associated with the eigenvalue $\lambda^{n}\in\sigma(\hat{L})$ as $\varepsilon \to 0$.
\end{proof}
\subsection{Discussion of Admissible Regularizations}
\label{sec:Appendix_admissible_regs}
In our considerations, the treatment of the singular Sturm-Liouville equation is directly linked to a differential geometric viewpoint as the divergent potential is interpreted as the curvature of a certain degenerate curve in the CPA. Thus, there are two convergence problems occurring simultaneously,
\begin{enumerate}
    \item the convergence of geometric objects, i.\,e. the point sets which are the traces of the curves with curvatures $\kappa_{\varepsilon}$, and
    \item the convergence of the associated operators $(\hat{L}_{\varepsilon}, \mathcal{D}(\hat{L}_{\varepsilon}))$.
\end{enumerate}
Under the assumptions in the above construction both the limiting trace of the curves and the limiting spectrum of the operators are unique and resemble the degenerate case that was to be approximated. So the assumptions listed in the theorems above define what we shall call an \textit{admissible regularization}.
\begin{definition}[Admissible Regularization]
    Let $J = (a,b) \subset \mathbb{R}$ be a bounded interval with $s\in J$ the arc length parameter. Let further be $\mathcal{X}_{\mathcal{M}}$ the parametrization of a degenerate planar curve such that its degenerate locus $\Xi_{2,1}(\mathcal{M)}$ is given by divergences of the curvature $\kappa\in L^{1}(J)$. Consider the special case that $\kappa$ behaves such that $v \coloneqq -\frac{1}{4}\kappa^{2}$ fulfills $v\in H^{-1}(J)$, i.\,e. there is a primitive $u\in L^{2}(J): v=u'$ in the distributional sense.\\[5pt]
    A family of regular curve parametrizations $\{\mathcal{X}_{\mathcal{M}_{\varepsilon}}\}_{\varepsilon\in\mathbb{R}^{+}}$ with curvature functions $\{ \kappa_{\varepsilon} \}_{\varepsilon\in\mathbb{R}^{+}} \subset C^{0}(J)$, yielding potential functions $v_{\varepsilon} \coloneqq -\frac{1}{4}\kappa_{\varepsilon}^{2}$ with classical antiderivatives $u_{\varepsilon}$, is called an \textit{admissible regularization of} $\mathcal{X}_{\mathcal{M}}$ if the following statements hold:
    \begin{enumerate}[label = (\roman*)]
        \item The curves converge in the geometric sense, so that
        \begin{align*}
            \left\vert \mathcal{X}_{\mathcal{M}_{\varepsilon}}(s) - \mathcal{X}_{\mathcal{M}}(s) \right\vert \to 0\ ,\ \ \left\vert \kappa_{\varepsilon}(s) - \kappa(s) \right\vert \to 0\ ,\quad (\varepsilon \to0)\quad \forall\ s\in J\setminus \mathcal{X}_{\mathcal{M}}^{-1}[\Xi_{2,1}(\mathcal{M})]\ ,
        \end{align*}
        i.\,e. the curve positions and all derived geometric quantities converge in the pointwise sense almost everywhere.
        \item The primitives of the curvature potentials converge in $(L^{2}(J), \langle \cdot, \cdot \rangle_{L^{2}(J)})$, so that
        \begin{align*}
            \Vert u_{\varepsilon} - u\Vert_{L^{2}(J)} \to 0\ ,\quad (\varepsilon \to 0)\ .
        \end{align*}
    \end{enumerate}
    A family $\{\mathcal{X}_{\mathcal{M}_{\varepsilon}}\}_{\varepsilon\in\mathbb{R}^{+}}$ for which only (i) is fulfilled is called a \textit{pre-admissible regularization}.
\end{definition}
At this point we want to discuss pre-admissible regularizations using curves that converge to the limiting curve in the geometrical sense, but whose operators in the effective CPA-Schrödinger equation do not converge in the appropriate sense. This means, in this case the manifolds and their geometric invariants such as curvature do converge as $(\mathcal{M}_{\varepsilon}, g_{\varepsilon}) \to (\mathcal{M}, g)$ and $\kappa_{\varepsilon} \to \kappa$ in the pointwise sense as $\varepsilon \to 0$, but the primitives $u_{\varepsilon}$ proportional to the anti-derivatives of the squares of the curvatures do not converge in the $L^{2}(J)$-sense so that the assumptions of theorem \ref{Th:ConvergenceEigenfunctions} are violated. Such non-admissible regularizations do exist and are characterized by an additive correction term as the following result shows. 
\begin{theorem}
    Let $J = (a,b) \subset \mathbb{R}$ be a finite interval and $\mathcal{X}_{\mathcal{M}}$ the arc length parametrization of a degenerate planar curve possessing a divergent curvature $\kappa$ (in $s_{0}\in J$) such that $v \coloneqq -\frac{\kappa^{2}}{4} \in H^{-1}(J)$ with primitive $u\in L^{2}(J)$. Let further $\left\{ \mathcal{X}_{\mathcal{M}_{\varepsilon}} \right\}_{\varepsilon\in\mathbb{R}^{+}}$ be a pre-admissible regularization of $\mathcal{X}_{\mathcal{M}}$. Then, the following statements are equivalent:
    \begin{enumerate}[label = (\roman*)]
        \item The primitives $\{ u_{\varepsilon}\}_{\varepsilon\in\mathbb{R}^{+}} \subset C^{1}(J)$ do not converge in the $L^{2}(J)$-sense, i.\,e.
        \begin{align*}
            \Vert u_{\varepsilon} - u\Vert_{L^{2}(J)} \not\to 0\ ,\quad (\varepsilon \to 0)\ .
        \end{align*}
        \item There is an admissible regularization with curves $(\mathcal{X}_{\tilde{\mathcal{M}_{\varepsilon}}})_{\varepsilon\in \mathbb{R}^{+}}$ with curvatures $\{ \tilde{\kappa}_{\varepsilon} \}_{\varepsilon\in \mathbb{R}^{+}}$ such that
        \begin{align*}
            \kappa_{\varepsilon}^{2} = \tilde{\kappa}_{\varepsilon}^{2} + \Delta_{\varepsilon}\ ,
        \end{align*}
        where $\Delta_{\varepsilon}\in C^{1}(J)$ is a function that vanishes almost everywhere in the pointwise sense,
        \begin{align*}
            \vert \Delta_{\varepsilon}(s)\vert \to 0\ ,\quad (\varepsilon\to 0)\ \ \forall\ s\in J\setminus \{ s_{0} \}
        \end{align*}
        but
        \begin{align*}
            U_{\varepsilon}(s) \coloneqq - \frac{1}{4}\,\int\limits_{a}^{s} \Delta_{\varepsilon}(s')\ \mathrm{d}s'
        \end{align*}
        does not vanish in the $L^{2}(J)$-sense,
        \begin{align*}
            \Vert U_{\varepsilon}\Vert_{L^{2}(J)} \not\to 0\ ,\quad (\varepsilon \to 0)\ .
        \end{align*}
    \end{enumerate}
\end{theorem}
\begin{proof}
    $\mathrm{(i) \Rightarrow \mathrm{(ii)}}$: Let $(\mathcal{X}_{\mathcal{M}_{\varepsilon}})_{\varepsilon \in \mathbb{R}^{+}}$ be a pre-admissible regularization such that there exists a constant $C \in \mathbb{R}^{+}$ so that
    \begin{align}
        \Vert u_{\varepsilon} - u \Vert_{L^{2}(J)} \to C\ ,\quad (\varepsilon \to 0)\ .
    \end{align}
    There exists an admissible regularization $(\mathcal{X}_{\tilde{\mathcal{M}}_{\varepsilon}})_{\varepsilon \in \mathbb{R}^{+}}$ that is constructed as follows: Let $B_{\varepsilon}(s_{0}) \subset J$ be an open ball with radius $\varepsilon$ whose center is the singular point $s_{0} \in \Xi_{2,1}(\mathcal{M})$, and $f_{\varepsilon} \in C^{0}(B_{\varepsilon}(s_{0}))$ a function such that $\vert f_{\varepsilon}(s)\vert \leq \vert \kappa(s) \vert$ for all $s\in B_{\varepsilon}(s_{0})$. With this, define the regularization $(\mathcal{X}_{\tilde{\mathcal{M}}_{\varepsilon}})_{\varepsilon \in \mathbb{R}^{+}}$ such that the curvature functions are given by
    \begin{align}
        \tilde{\kappa}_{\varepsilon}(s) \coloneqq \begin{cases}
            \kappa(s) &: s\in J \setminus B_{\varepsilon}(s_{0}) \\
            f_{\varepsilon}(s) &: s\in B_{\varepsilon}(s_{0})
        \end{cases}\ .
    \end{align}
    Note that such a function $f_{\varepsilon}$ always exists due to the divergence of $\kappa$; for example, for sufficiently small $\varepsilon\in \mathbb{R}^{+}$ consider the linear function defined by the values of $\kappa$ at the limiting points $\partial B_{\varepsilon}(s_{0}) = \left\{ s_{0} - \varepsilon\,;\, s_{0} + \varepsilon \right\}$. Then, it holds $\tilde{\kappa}_{\varepsilon}^{2}(s) \leq \kappa^{2}(s)$ for all $s\in J$ and the monotonicity of the Lebesgue integral yields
    \begin{align}
        0 \leq \int\limits_{J} \tilde{\kappa}_{\varepsilon}^{2}(s)\ \mathrm{d}s \leq \int\limits_{J} \kappa^{2}(s)\ \mathrm{d}s\ .
    \end{align}
    Obviously, the pointwise convergence
    \begin{align}
        \vert \tilde{\kappa}_{\varepsilon}(s) - \kappa(s) \vert \to 0\ ,\quad (\varepsilon \to 0)\quad \forall\ s\in J\setminus\{ s_{0} \}
    \end{align}
    holds, so that $(\mathcal{X}_{\tilde{\mathcal{M}}_{\varepsilon}})_{\varepsilon \in \mathbb{R}^{+}}$ is a pre-admissible regularization. Furthermore, this means that convergence holds almost everywhere on $J$, so the monotone convergence theorem implies
    \begin{align}
       \Vert \tilde{u}_{\varepsilon} - u \Vert_{L^{2}(J)} &= \left\Vert \int \left( \tilde{v}_{\varepsilon}(s') - v(s') \right)\ \mathrm{d}s'  \right\Vert_{L^{2}(J)} = \left\Vert -\frac{1}{4}\,\int\limits_{J} \left( \tilde{\kappa}^{2}_{\varepsilon}(s) - \kappa^{2}(s) \right)\ \mathrm{d}s \right\Vert_{L^{2}(J)} \to 0\ ,\quad (\varepsilon \to 0)\ .
    \end{align}
    Thus, $(\mathcal{X}_{\tilde{\mathcal{M}}_{\varepsilon}})_{\varepsilon \in \mathbb{R}^{+}}$ is an admissible regularization.\\
    Then, define
    \begin{align}
        \Delta_{\varepsilon} \overset{:}{=} \kappa_{\varepsilon}^{2} - \tilde{\kappa}_{\varepsilon}^{2}\ \Leftrightarrow\ \kappa_{\varepsilon}^{2} = \tilde{\kappa}_{\varepsilon}^{2} + \Delta_{\varepsilon}\ \Leftrightarrow\ v_{\varepsilon} = \tilde{v}_{\varepsilon} - \frac{1}{4}\,\Delta_{\varepsilon}\ .
    \end{align}
    Obviously, $\Delta_{\varepsilon}\in C^{0}(J)$ for all $\varepsilon \in \mathbb{R}^{+}$, and there is a potential primitive proportional to its antiderivative
    \begin{align}
        U_{\varepsilon}(s) \coloneqq -\frac{1}{4}\,\int\limits_{a}^{s}\Delta_{\varepsilon}(s')\ \mathrm{d}s'\ .
    \end{align}
    From the above definition, one directly computes that the $\Delta_{\varepsilon}$ vanish in the pointwise sense,
    \begin{align}
        0 \leq \vert \Delta_{\varepsilon}(s)\vert  = \left\vert \kappa_{\varepsilon}^{2}(s) - \tilde{\kappa}_{\varepsilon}^{2}(s) \right\vert \leq \left\vert \kappa_{\varepsilon}^{2}(s) - \kappa^{2}(s)  \right\vert + \left\vert \tilde{\kappa}_{\varepsilon}^{2}(s) - \kappa^{2}(s)  \right\vert \to 0\ ,\quad (\varepsilon\to 0) \quad \forall\ s\in J\setminus\{ s_{0} \}\ ,
    \end{align}
    but their corresponding potential primitives do not vanish in the $L^{2}(J)$-sense,
    \begin{align}
        \begin{aligned}
        \Vert U_{\varepsilon} \Vert_{L^{2}(J)} &= \left\Vert -\frac{1}{4}\,\int \Delta_{\varepsilon}(s)'\ \mathrm{d}s'  \right\Vert_{L^{2}(J)} = \left\Vert -\frac{1}{4}\,\int \left( \kappa_{\varepsilon}^{2}(s') - \tilde{\kappa}_{\varepsilon}^{2}(s') \right)\ \mathrm{d}s'  \right\Vert_{L^{2}(J)} = \left\Vert \int \left( v_{\varepsilon}(s') - \tilde{v}_{\varepsilon}(s') \right)\ \mathrm{d}s'  \right\Vert_{L^{2}(J)} \\[5pt]
        &= \left\Vert  u_{\varepsilon} - \tilde{u}_{\varepsilon} \right\Vert_{L^{2}(J)} \geq \left\vert \Vert u_{\varepsilon} - u \Vert_{L^{2}(J)} - \Vert \tilde{u}_{\varepsilon} - u \Vert_{L^{2}(J)} \right\vert \to C\ ,\quad (\varepsilon \to 0)\ .
        \end{aligned}
    \end{align}
    \par
    $\mathrm{(ii) \Rightarrow \mathrm{(i)}}$: The proof follows by direct computation. Let the admissible regularization $(\mathcal{X}_{\tilde{\mathcal{M}}_{\varepsilon}})_{\varepsilon \in \mathbb{R}^{+}}$ be given as described. Then, because of
    \begin{align}
        \kappa_{\varepsilon}^{2} = \tilde{\kappa}_{\varepsilon}^{2} + \Delta_{\varepsilon}\ \Leftrightarrow\ v_{\varepsilon} = \tilde{v}_{\varepsilon} - \frac{1}{4}\,\Delta_{\varepsilon}\ \Leftrightarrow\ u_{\varepsilon} = \tilde{u}_{\varepsilon} + U_{\varepsilon} + C\ ,\ C\in\mathbb{R}\ ,
    \end{align}
    it holds that there is a constant $D\in \mathbb{R}^{+}$ such that
    \begin{align}
        \Vert u_{\varepsilon} - u\Vert_{L^{2}(J)} = \Vert \tilde{u}_{\varepsilon} + U_{\varepsilon} + C - u\Vert_{L^{2}(J)} \geq \left\vert \Vert \tilde{u}_{\varepsilon} - u \Vert_{L^{2}(J)} - \Vert U_{\varepsilon} + C \Vert_{L^{2}(J)}  \right\vert \to D\ ,\quad (\varepsilon \to 0)\ .
    \end{align}
\end{proof}
This result illustrates that using a regularization of curves that merely converge in the geometric sense, such that $\kappa_{\varepsilon} \to \kappa$ in a pointwise manor, is insufficient for our formulation of quantum mechanics. Instead, the $L^{2}(J)$-convergence of the primitives is necessary to correctly compute the limiting spectrum via this regularization approach.
\end{widetext}
\end{appendices}
%
%%%%%%%%%%%%%%%%%%%%%%%%%%%%%%%%%%%%%%%%%%%%%%%%%%%%%%%%%%%%%%%%%%%%%%%%%%%%%%%%%%%%%%%%%%%%%%%%%%%%%%%%%%%%%%%%%%%%%%%%%%%%%%%%%%%
%%%%%%%%%%%%%%%%%%%%%%%%%%%%%%%%%%%%%%%%%%%%%%%%%%%%%%%%%%%%%%%%%%%%%%%%%%%%%%%%%%%%%%%%%%%%%%%%%%%%%%%%%%%%%%%%%%%%%%%%%%%%%%%%%%%
% Acknowledgments
\section*{Acknowledgment}
The authors express their thank to Prof. Dr. habil. Tom\'{a}\v{s} Dohnal, Dr. Mathias Schäffner, and Dr. Nora Christine Doll (currently all at MLU, Institut für Mathematik) for valuable discussions on operator theory during the research process.
%
%%%%%%%%%%%%%%%%%%%%%%%%%%%%%%%%%%%%%%%%%%%%%%%%%%%%%%%%%%%%%%%%%%%%%%%%%%%%%%%%%%%%%%%%%%%%%%%%%%%%%%%%%%%%%%%%%%%%%%%%%%%%%%%%%%%
%%%%%%%%%%%%%%%%%%%%%%%%%%%%%%%%%%%%%%%%%%%%%%%%%%%%%%%%%%%%%%%%%%%%%%%%%%%%%%%%%%%%%%%%%%%%%%%%%%%%%%%%%%%%%%%%%%%%%%%%%%%%%%%%%%%
% Literature
%\FloatBarrier
\bibliography{bib_vertex}
\end{document}